\documentclass[runningheads]{llncs}
\usepackage[T1]{fontenc}
\usepackage{graphicx}
\usepackage{color}
\usepackage{amsmath,amssymb,amsfonts}
\usepackage{mathtools}
\usepackage{stmaryrd}
\usepackage{hyperref}
\usepackage{xspace}

\usepackage{algorithm}
\usepackage[noend]{algpseudocode}
\usepackage{changepage}
\usepackage{paralist}
\usepackage{thm-restate}
\usepackage{multirow}
\usepackage[table]{xcolor}
\usepackage{cleveref}
\usepackage{csquotes}

\usepackage[
colorinlistoftodos, 
textwidth=\marginparwidth, 
textsize=scriptsize, 
]{todonotes}

\usepackage{mdframed}

\usepackage{pgf}
\usepackage{tikz}
\usetikzlibrary{shapes,shapes.geometric,fit,trees,decorations,arrows,automata,shadows,positioning,plotmarks,backgrounds,tikzmark}
\usetikzlibrary{calc,matrix,fit,petri,decorations.markings,decorations.pathmorphing,patterns,intersections,decorations.text}

\newboolean{changeLTLnotation}
\setboolean{changeLTLnotation}{false}
\newboolean{changePowerNotation}
\setboolean{changePowerNotation}{false}

\newcommand{\p}[1]{\ensuremath{\text{Player}~#1}}

\newcommand{\toolname}{\textsc{PeSTel}\xspace}

\newcommand{\genziel}{\textsc{GenZiel}\xspace}

\newcommand{\N}{\mathbb{N}}

\newcommand{\bigO}{\mathcal{O}}
\newcommand{\cupdot}{\mathbin{\mathaccent\cdot\cup}}

\ifthenelse{\boolean{changePowerNotation}}{
	P
}{
}

\newcommand{\set}[1]{\left\lbrace #1\right\rbrace}
\newcommand{\tup}[1]{\left( #1\right)}

\newcommand{\inodd}{\in_{\mathrm{odd}}}
\newcommand{\ineven}{\in_{\mathrm{even}}}

\newcommand{\gamegraph}{\ensuremath{G}}

\newcommand{\priority}{\ensuremath{\mathbb{P}}}

\newcommand{\paritygame}{\mathit{Parity}}

\newcommand{\assump}{\ensuremath{\Psi}}
\newcommand{\assumpsafe}{\ensuremath{\assump_{\textsc{unsafe}}}}

\newcommand{\assumpgrlive}{\ensuremath{\assump_{\textsc{live}}}}

\newcommand{\assumpdep}{\ensuremath{\assump_{\textsc{colive}}}}

\newcommand{\livegroup}{\mathcal{H}}
\newcommand{\livegroupSingle}{H_i}
\newcommand{\livegroupSingleN}{H}
\newcommand{\colivegroup}{D}
\newcommand{\safegroup}{S}

\newcommand{\faulty}{F}

\newcommand{\safeTemp}{\textsc{SafetyTemplate}}
\newcommand{\reachTemp}{\textsc{ReachTemplate}}
\newcommand{\buchiTemp}{\textsc{B\"uchiTemplate}}
\newcommand{\coBuchiTemp}{\textsc{coB\"uchiTemplate}}
\newcommand{\parityTemp}{\textsc{ParityTemplate}}
\newcommand{\extract}{\textsc{ExtractStrategy}}
\newcommand{\checkconflictfree}{\textsc{CheckTemplate}}
\newcommand{\PUFcorrection}{\textsc{FaultCorrection}}

\newcommand{\solveBuchi}{\textsc{B\"uchi}}
\newcommand{\solveCobuchi}{\textsc{CoB\"uchi}}

\newcommand{\solveSafety}{\textsc{Safety}}

\newcommand{\edges}{\textsc{Edges}}

\newcommand{\pre}[2]{\textsf{pre}\ensuremath{_{#1}(#2)}}
\newcommand{\upre}[2]{\textsf{upre}\ensuremath{_{#1}(#2)}}
\newcommand{\cpre}[3]{\textsf{cpre}\ensuremath{^{#3}_{#1}(#2)}} 
\newcommand{\cprea}[2]{\textsf{cpre}\ensuremath{^a_{#1}(#2)}}

\newcommand{\cprez}[2]{\textsf{cpre}\ensuremath{^0_{#1}(#2)}}

\newcommand{\uattr}[2]{\textsf{uattr}\ensuremath{_{#1}(#2)}}
\newcommand{\attra}[2]{\textsf{attr}\ensuremath{^a_{#1}(#2)}}
\newcommand{\attro}[2]{\textsf{attr}\ensuremath{^1_{#1}(#2)}}
\newcommand{\attrz}[2]{\textsf{attr}\ensuremath{^0_{#1}(#2)}}

\newcommand{\stratFollows}{\Vdash}

\ifthenelse{\boolean{changeLTLnotation}}{
	\renewcommand{\bigcirc}{\mathbf{X}}
	\renewcommand{\lozenge}{\mathbf{E}}
	\renewcommand{\square}{\mathbf{G}}
}{
}

\newcommand{\game}{\mathcal{G}}

\newcommand{\win}{\mathcal{W}}
\newcommand{\wino}{\win_0}
\newcommand{\winl}{\win_1}
\newcommand{\wini}{\win_i}
\newcommand{\lang}{\mathcal{L}}

\newcommand{\pz}{\p{0}}
\newcommand{\po}{\p{1}}

\newcommand{\vertex}{V}
\newcommand{\vertexz}{\vertex^0}
\newcommand{\vertexo}{\vertex^1}
\newcommand{\vertexi}{\vertex^i}

\newcommand{\spec}{\Phi}
\newcommand{\edge}{E}
\newcommand{\strat}{\pi}
\newcommand{\stratz}{\pi^0}
\newcommand{\strato}{\pi^1}
\newcommand{\strati}{\pi^i}

\newcommand{\play}{\rho}

\makeatletter
\newsavebox{\@brx}
\newcommand{\llangle}[1][]{\savebox{\@brx}{\(\m@th{#1\langle}\)}%
	\mathopen{\copy\@brx\kern-0.5\wd\@brx\usebox{\@brx}}}
\newcommand{\rrangle}[1][]{\savebox{\@brx}{\(\m@th{#1\rangle}\)}%
	\mathclose{\copy\@brx\kern-0.5\wd\@brx\usebox{\@brx}}}

\makeatother

\newcommand{\src}{\mathit{src}}

\newcommand{\buchi}{\ifmmode B\ddot{u}chi \else B\"uchi \fi}
\newcommand{\cobuchi}{\ifmmode co\text{-}B\ddot{u}chi \else co-B\"uchi \fi}

\makeatletter 
\newif\ifFIRST
\newif\ifSECOND
\let\LISTOP\relax
\newcommand{\List}[4][\;]{#3#1%
	\FIRSTtrue
	\@for\i:=#2\do{%
		\ifFIRST\LISTOP{\i}\FIRSTfalse\else,\LISTOP{\i}\fi%
	}%
	#1#4%
	\let\LISTOP\relax
}
\makeatother

\makeatletter

\newcommand{\Set}[2][]{\List[#1]{#2}{\{}{\}}}
\newcommand{\VSet}[2][]{\let\LISTOP\val\List[#1]{#2}{\{}{\}}}

\newcommand{\VTuple}[2][]{\let\LISTOP\val\List[#1]{#2}{(}{)}}

\renewenvironment{proof}{\noindent{\color{darkgray}\textbf{Proof.}\ignorespaces }}{\hfill $\blacksquare$ \vspace*{1em}\newline} 
 
\newenvironment{claimproof}{\begin{adjustwidth}{2em}{}
		{\color{darkgray}\emph{Proof.}\ignorespaces }
	}{\hfill $\vartriangleleft$ \vspace*{0.5em}\end{adjustwidth}}

\newcommand{\abs}[1]{\left\lvert #1 \right\rvert}

\newcommand{\compose}{\textsc{ComposeTemplate}}

\newcommand{\conflict}{\mathcal{C}}

\usepackage{expl3}[2012-07-08]

\colorlet{darkgreen}{green!80!black}
\colorlet{darkred}{red!80!black}
\usetikzlibrary{arrows, automata, shapes}
\tikzset{auto, >= stealth}
\tikzset{every edge/.append style={thick, shorten >= 1pt}}
\tikzset{initial/.style={draw, thick, <-, shorten <=1pt}}
\tikzset{player0/.style = {draw, thick, shape=circle, minimum size=5mm}}
\tikzset{player1/.style = {draw, thick, shape=rectangle, minimum size=5mm}}
\tikzset{bplayer0/.style = {draw, thick, shape=ellipse, minimum size=5mm,text width=1.1cm}}
\tikzset{bplayer1/.style = {draw, thick, shape=rectangle, minimum size=5mm,text width=1.6cm}}
\newcommand\pos{1.4}

\begin{document}
\title{Synthesizing Permissive Winning Strategy Templates for Parity Games}
%
%

\author{Ashwani Anand\inst{1} \and Satya Prakash Nayak\inst{1} \and \\ Anne-Kathrin Schmuck\inst{1}}

\authorrunning{A. Anand, S. P. Nayak, and A. Schmuck}
%
\institute{Max Planck Institute for Software Systems, Kaiserslautern, Germany\\
	\email{\{ashwani,sanayak,akschmuck\}@mpi-sws.org}\\
}
\maketitle              
\begin{abstract}
We present a novel method to compute \emph{permissive winning strategies} in two-player games over finite graphs with $ \omega $-regular winning conditions. Given a game graph $G$ and a parity winning condition $\spec$, we compute a 
\emph{winning strategy template} $\assump$ that collects an infinite number of winning strategies for objective $\spec$ in a concise data structure. We use this new representation of sets of winning strategies to tackle two problems arising from applications of two-player games in the context of cyber-physical system design -- (i) \emph{incremental synthesis}, i.e., adapting strategies to newly arriving, \emph{additional} $\omega$-regular objectives $\spec'$, and (ii) \emph{fault-tolerant control}, i.e., adapting strategies to the occasional or persistent unavailability of actuators. 
The main features of our strategy templates -- which we utilize for solving these challenges -- are their easy computability, adaptability, and compositionality. For \emph{incremental synthesis}, we empirically show on a large set of benchmarks that our technique vastly outperforms existing approaches if the number of added specifications increases. 
While our method is not complete, 
our prototype implementation returns the full winning region in all 1400 benchmark instances, i.e.\ handling a large problem class efficiently in practice.

\end{abstract}
%
%
%
\section{Introduction}
Two-player $\omega$-regular games on finite graphs are an established modeling and solution formalism for many challenging problems in the context of correct-by-construction cyber-physical system (CPS) design \cite{tabuada2009verification,alur2015principles,belta2017formal}. Here, control software actuating a technical system \enquote{plays} against the physical environment. The winning strategy of the system player in this two-player game results in software which ensures that the controlled technical system fulfills a given temporal specification for any (possible) event or input sequence generated by the environment. 
Examples include warehouse robot coordination \cite{ScherKressGazit_2020}, reconfigurable manufacturing systems \cite{lesi2016towards}, and adaptive cruise control \cite{nilsson2015correct}.
In these applications, the technical system under control, as well as its requirements, are developing and changing during the design process. 
It is therefore desirable to allow for maintainable and adaptable control software. This, in turn, requires solution algorithms for two-player $\omega$-regular games which allow for this adaptability.

This paper addresses this challenge by providing a new algorithm to efficiently compute \emph{permissive winning strategy templates} in parity games which enable rich \emph{strategy adaptations}. Given a game graph $G=(V,E)$ and an objective $\spec$ a winning strategy template $\assump$ characterizes the winning region $\mathcal{W}\subseteq V$ along with three types of local edge conditions -- a \emph{safety}, a \emph{co-live}, and a \emph{live-group} template. The conjunction of these basic templates allows us to capture infinitely many winning strategies over $G$ w.r.t.\ $\spec$ in a simple data structure that is both (i) easy to obtain during synthesis, and (ii) easy to adapt and compose.

We showcase the usefulness of \emph{permissive winning strategy templates} in the context of CPS design by two application scenarios: (i) \emph{incremental synthesis}, where strategies need to be adapted to newly arriving \emph{additional} $\omega$-regular objectives $\spec'$, and (ii) \emph{fault-tolerant control}, where strategies need to be adapted to the occasional or persistent unavailability of actuators, i.e., system player edges. 

We have implemented our algorithms in a prototype tool \toolname and run it on more than $1400$ benchmarks adapted from the SYNTCOMP benchmark suite \cite{benchmark:syntcomp}. These experiments show that our class of templates effectively avoids re-computations for the required strategy adaptations. For \emph{incremental synthesis}, our experimental results are previewed in \cref{fig:intro-experiement-sensitivity}, where we compare \toolname against the state-of-the-art solver \genziel~\cite{chatterjee2007:generalizedparitygames} for generalized parity objectives, i.e., finite conjunction of parity objectives. We see that \toolname is as efficient as \genziel whenever all conjuncts of the objective are given \emph{up-front} (\cref{fig:intro-experiement-sensitivity} (left)) - even outperforming it in more than 90\% of the instances. 
Whenever conjuncts of the objective arrive \emph{one at a time}, \toolname outperforms the existing approaches significantly if the number of objectives increases (\cref{fig:intro-experiement-sensitivity} (right)). 
This shows the potential of \toolname towards more adaptable and maintainable control software for CPS.

\begin{figure}[t]
    \includegraphics[scale=0.29]{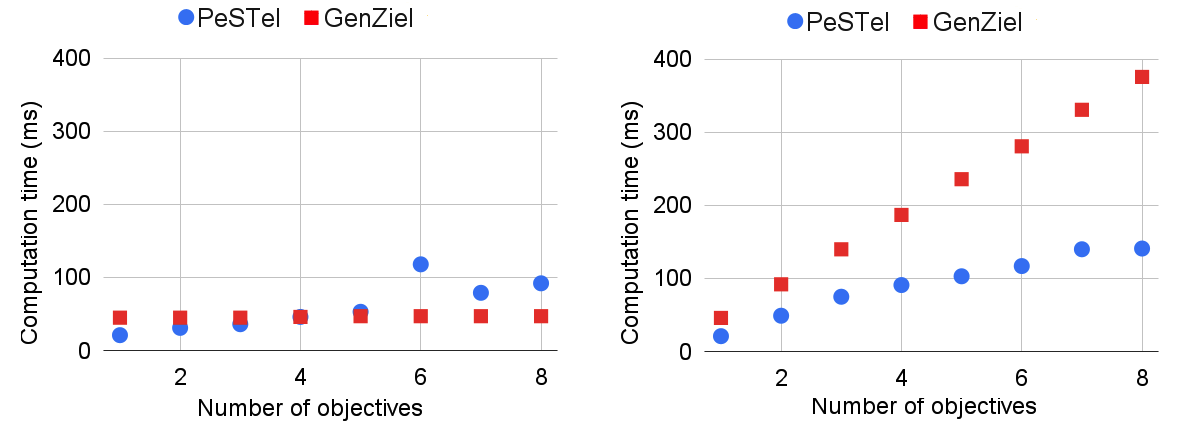}
    \vspace{-0.2cm}
     \caption{Experimental results over 1400 generalized parity games comparing the performance of our tool \toolname against the state-of-the-art generalized parity solver \genziel \cite{chatterjee2007:generalizedparitygames}. Data points give the average execution time (in ms) over all instances with the same number of parity objectives. Left: all objectives are given \emph{upfront}. Right: objectives are added \emph{one-by-one}. See \cref{sec:experiments} for more details on those experiments.}
     \label{fig:intro-experiement-sensitivity}
     \vspace{-0.5cm}
 \end{figure}

\smallskip
\noindent\textbf{Illustrative Example.}
To appreciate the simplicity and easy adaptability of our strategy templates, consider the game graph in \cref{fig:introExample} (left). The first winning condition $\spec_1$ requires vertex $f$ to never be seen along a play. This can be enforced by $\p{0}$ from vertices $\mathcal{W}_0=\{a,b,c,d\}$ called the \emph{winning region}. The safety template $\assump_1$ ensures that the game always stays in $\mathcal{W}_0$ by forcing the edge $e_{de}$ to never be taken. It is easy to see that every $\p{0}$ strategy that follows this rule results in plays which are winning if they start in $\mathcal{W}_0$. Now consider the second winning condition $\spec_2$ which 
requires vertex $c$ or $d$ to be seen infinitely often. This induces the live-group template $\assump_2$ which requires that whenever vertex $a$ is seen infinitely often, either edge $e_{ac}$ or edge $e_{ad}$ needs to be taken infinitely often. It is easy to see that any strategy that complies with this edge-condition is winning for $\p{0}$ from every vertex and there are infinitely many such compliant winning strategies. Finally, we consider condition $\spec_3$ requiring vertex $b$ to be seen only finitely often. This induces the strategy template $\assump_3$ which is a co-liveness template requiring that all edges from $\p{0}$ vertices which unavoidably lead to $b$ (i.e., $e_{ab}$, $e_{bd}$, and $e_{de}$) are taken only finitely often. We can now combine all templates into a new template $\assump'=\assump_1\wedge\assump_2\wedge\assump_3$ and observe that all strategies compliant with $\assump'$ are winning for $\spec'=\spec_1\wedge\spec_2\wedge\spec_3$.

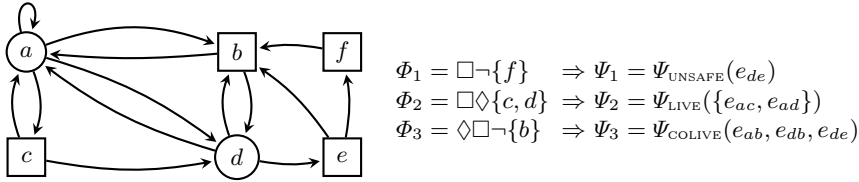
\begin{figure}[t]
\vspace{-0.5cm}
\centering
 \begin{tikzpicture}
	  		\node[player0] (0) at (0, 0) {$a$};
	  		\node[player1] (1) at (2*\fpeval{\pos}, 0) {$b$};
			\node[player1] (2) at (0,-\fpeval{\pos}) {$c$};
			\node[player0] (3) at (2*\fpeval{\pos}, -\fpeval{\pos}) {$d$};
			\node[player1] (4) at (3*\fpeval{\pos}, -\fpeval{\pos}) {$e$};
			\node[player1] (5) at (3*\fpeval{\pos}, 0) {$f$};
			\path[->] (0) edge[loop above] () edge[bend left = 20] (1) edge[bend left = 20] (2) edge[bend left = 10] (3);
			\path[->] (1) edge[bend left =5] (0) edge[bend left = 20] (3);
			\path[->] (2) edge[bend left = 20] (0) edge[bend right = 10] (3);
			\path[->] (3) edge[bend left = 10] (0) edge[bend left = 20] (1) edge[bend right = 10] (4);
			\path[->] (4) edge[bend right = 10] (1) edge[bend right = 10] (5);
			\path[->] (5) edge[bend right = 10] (1);

\node[text width=8cm,anchor=north west] at (4.8,0) {
$\spec_1=\Box\neg\{f\}$ \\
$\spec_2=\Box\Diamond\{c,d\}$ \\
$\spec_3=\Diamond\Box\neg\{b\}$

}; 
\node[text width=8cm,anchor=north west] at (7,0) {
$\Rightarrow$ $\assump_1=\assumpsafe(e_{de})$\\
$\Rightarrow$ $\assump_2=\assumpgrlive(\{e_{ac},e_{ad}\})$\\
$\Rightarrow$ $\assump_3=\assumpdep(e_{ab},e_{db},e_{de})$
};
 \end{tikzpicture}
\vspace{-0.6cm}
 \caption{A two-player game graph with $\po$ (squares) and $\pz$ (circles) vertices, different winning conditions $\spec_i$, and corresponding winning strategy templates $\assump_i$.}\label{fig:introExample}
\vspace{-0.5cm}
\end{figure}

In addition to their compositionality, strategy templates also allow for local strategy adaptations in case of edge unavailability faults. Consider again the game in \cref{fig:introExample} with the objective $ \spec_2 $. 
Suppose that $ \pz $ follows the strategy $ \strat $: $ a\mapsto d $ and $ d\mapsto a $, which is compliant with $ \assump_2 $. If the edge $ e_{ad} $ becomes unavailable, we would need to re-solve the game for the modified game graph $\gamegraph'=(V,E\setminus\{e_{ad}\})$. However, given the strategy template $ \assump_2 $ we see that the strategy  $ \strat' $:  $ a\mapsto c $ and $ d\mapsto a $ is actually compliant with $\assump_2$ over $\gamegraph'$. This allows us to obtain a new strategy without re-solving the game. 

While these examples demonstrate the potential of templates for strategy adaptation, there exist scenarios where conflicts between templates or graph modifications arise, 
which require re-computations. Our empirical results, however, show that such conflicts rarely appear in practical benchmarks. This suggests that our technique can handle a large problem class efficiently in practice. 


\smallskip
\noindent\textbf{Related work.}
The class of templates we use was introduced in \cite{permissiveassumptions} and utilized to represent environment assumptions that enable a system to fulfill its specifications in a cooperative setting. Contrary to  \cite{permissiveassumptions}, this paper uses the same class of templates to represent the system's winning strategies in a zero-sum setting.

While the computation of \emph{permissive strategies} for the control of CPS is an established concept in the field of supervisory control \footnote{See \cite{EhlersJdeds,SchmuckJdeds20,MajumdarSchmuck_23_TAC} for connections between supervisory control and reactive synthesis.} \cite{CassandrasLafortune3rd,wonhamcai_book}, it has also been addressed in reactive synthesis where the considered specification class is typically more expressive, 
e.g., Bernet et al. \cite{bernet2002:permissiveStrategies} introduce permissive strategies that encompass all the behaviors of positional strategies and Neider et al. \cite{neider2014:mullergamestosafety} introduce permissiveness to subsume strategies that visit losing loops at most twice. Finally, Bouyer et al. \cite{bouyer2011:measuringpermissiveness} take a quantitative approach to measure the permissiveness of strategies, by minimizing the penalty of not being permissive. However, all these approaches are not optimized towards strategy adaptation and thereby typically fail to preserve enough behaviors to be able to effectively satisfy subsequent objectives. 
A notable exception is a work by Baier et al. \cite{Klein2015:mostGeneralController}. 
While their strategy templates are more complicated and more costly to compute than ours, they are \emph{maximally} permissive (i.e., capture \emph{all} winning strategies in the game). However, 
when composing multiple objectives, they restrict templates substantially which eliminates many compositional solutions that our method retains. 
This results in higher computation times and lower result quality for incremental synthesis compared to our approach.
As no implementation of their method is available, we could not compare both approaches empirically. 

Even without the incremental aspect, synthesizing winning strategies for conjunctions of $ \omega $-regular objectives is known to be a hard problem -- Chatterjee et al. \cite{chatterjee2007:generalizedparitygames} prove that the conjunction of even \emph{two} parity objectives makes the problem NP-complete. They provide a generalization of Zielonka's algorithm, called \genziel for generalized parity objectives (i.e., finite conjunction of parity objectives) which is compared to our tool \toolname in \cref{fig:intro-experiement-sensitivity}.
While \toolname is (in contrast to \genziel) not complete  --- i.e., there exist realizable synthesis problems for which \toolname returns no solution --- 
our prototype implementation returns the full winning region in all $1400$ benchmark instances. 


Fault-tolerant control is a well-established topic in control engineering \cite{faulttolerancebook}, with recent emphasis on the logical control layer \cite{Moor16,FritzZhang18}. While most of this work is conducted in the context of supervisory control, there are also some approaches in reactive synthesis. While \cite{neider2020synthesizing,MeiraGoesKangLafortuneTripakis_22_wodes} considers the \emph{addition} of \enquote{disturbance edges} and synthesizes a strategy that tolerates as many of them as possible, we look at the complementary problem, where edges, in particular system-player edges, disappear. To the best of our knowledge, the only algorithm that is able to tackle this problem without re-computation considers Büchi games \cite{ChatterjeeHenzinger_14_Buechi}. In contrast, our method is applicable to the more expressive class of Parity games.

\section{Preliminaries}\label{sec:prelim}
\noindent\textbf{Notation.}
We use $\mathbb{N}$ to denote the set of natural numbers including zero.
Given two natural numbers $a,b\in\mathbb{N}$ with $a<b$, we use $[a;b]$ to denote the set $\set{n\in\mathbb{N} \mid a\leq n\leq b}$.
For any given set $[a;b]$, we write $i\ineven [a;b]$ and $i\inodd [a;b]$ as shorthand for $i\in [a;b]\cap \set{0,2,4,\ldots}$ and $i\in [a;b]\cap \set{1,3,5,\ldots}$ respectively.
Given two sets $A$ and $B$, a relation $R\subseteq A\times B$, and an element $a\in A$, we write $R(a)$ to denote the set $\set{b\in B\mid (a,b)\in R}$.

\smallskip
\noindent\textbf{Languages.}
Let $\Sigma$ be a finite alphabet.
The notation $\Sigma^*$ and $\Sigma^\omega$ respectively denote the set of finite and infinite words over $\Sigma$, and $\Sigma^\infty$ is equal to $\Sigma^*\cup \Sigma^\omega$.
For any word $w\in \Sigma^\infty$, $w_i$ denotes the $i$-th symbol in $w$.
Given two words $u\in \Sigma^*$ and $v\in \Sigma^\infty$, the concatenation of $u$ and $v$ is written as the word $uv$.

\smallskip
\noindent\textbf{Game Graphs.}
A \emph{game graph} is a tuple $\gamegraph= \tup{V=V^0\cupdot V^1,E}$ where $(V,E)$ is a finite directed graph with \emph{vertices} $ V $ and \emph{edges} $ E $, and 
$ \vertexz, \vertexo\subseteq V$ form a partition of $V$. Without loss of generality, we assume that for every $v\in V$ there exists $v'\in V$ s.t.\ $(v,v')\in E$. 
A \emph{play} originating at a vertex $v_0$ is a finite or infinite sequence of vertices $\rho=v_0v_1\ldots \in V^\infty$.

\smallskip
\noindent\textbf{Winning Conditions/Objectives.}
Given a game graph $\gamegraph$, we consider winning conditions/objectives specified using a formula $\spec$ in \emph{linear temporal logic} (LTL) over the vertex set $V$, that is, we consider LTL formulas whose atomic propositions are sets of vertices $V$. 
In this case the set of desired infinite plays is given by the semantics of $\spec$ which is an $\omega$-regular language $\lang(\spec)\subseteq V^\omega$. 
Every game graph with an arbitrary $\omega$-regular set of desired infinite plays can be reduced to a game graph (possibly with a different set of vertices) with an LTL winning condition, as above.
The standard definitions of $\omega$-regular languages and LTL are omitted for brevity and can be found in standard textbooks \cite{baierbook}. To simplify notation we use $ e=(u,v) $ in LTL formulas as syntactic sugar for $ u\wedge \bigcirc v $, with $ \bigcirc$ as the LTL \emph{next} operator. We further use a set of edges $E' = \set{e_i}_{i\in [0;k]}$ as atomic proposition to denote $\bigvee_{i\in [0;k]} e_i$. 

\smallskip
\noindent\textbf{Games and Strategies.} A \emph{two-player (turn-based) game} is a pair $\game=\tup{\gamegraph,\spec}$ where $G $ is a game graph and 
$ \spec $ is a \emph{winning condition} over $\gamegraph$.
A strategy of $\p{i},~i\in\{0,1\}$, is a function $\strati\colon \vertex^*\vertexi\to \vertex$ such that for every $\play v \in \vertex^*\vertexi$ 
holds that $\strati(\play v)\in \edge(v)$. 
Given a strategy $\strati$, we say that the play $\play=v_0v_1\ldots$ is \emph{compliant} with $\strati$ if $v_{k-1}\in \vertexi$ implies $v_{k} = \strati(v_0\ldots v_{k-1})$ for all $k$.
We refer to a play compliant with $\strati$ and a play compliant with both $\stratz$ and $\strato$ as a \emph{$ \strati $-play} and a \emph{$ \stratz\strato $-play}, respectively. 
We collect all plays originating in a set $ S $ and compliant with $\strati$, (and compliant with both $\stratz$ and $\strato$) in the sets $\lang(S,\strati)$ (and $\lang(S,\stratz\strato)$, respectively). When $ S=V $, we drop the mention of the set in the previous notation, and when $ S $ is singleton $ \{v\} $, we simply write $ \lang(v,\strati) $ (and $ \lang(v,\stratz\strato) $, respectively). 

\smallskip
\noindent\textbf{Winning.}
Given a game $\game=(\gamegraph,\spec)$, a play $ \play $ in $ \game $ is \emph{winning for $ \p{0} $}, if $ \play\in\lang(\spec) $, and it is winning for $ \p{1} $, otherwise. A strategy $\strati$ for $\p{i}$ is \emph{winning from a vertex $ v\in V $} if all plays compliant with $ \strati $ and originating from $ v $ are winning for $ \p{i} $. 
We say that a vertex $v\in V$ is  \emph{winning for $\p{i}$}, if there exists a winning strategy $\strati$ from $ v $. 
We collect all winning vertices of $\p{i}$ in the \emph{$\p{i}$ winning region} $\wini\subseteq V$.
We always interpret winning w.r.t.\ $\p{0}$ if not stated otherwise.

%
%
\smallskip
\noindent\textbf{Strategy Templates.}
%
Let $ \stratz $ be a $ \p{0} $ strategy and $ \spec $ be an LTL formula. Then we say $ \stratz $ \emph{follows} $ \spec $, denoted $ \stratz\stratFollows\spec $, if for all $ \stratz $-plays $ \play $, $ \play$ belongs to $\lang(\spec) $, i.e. $ \lang(\stratz)\subseteq \lang(\spec) $. 
%
We refer to a set $ \assump=\set{\assump_1,\ldots, \assump_k} $ of LTL formulas as \emph{strategy templates} representing the set of strategies that follows $ \assump_1\wedge\ldots\wedge\assump_k $. We say a strategy template $\assump$ is \emph{winning from a vertex $v$} for a game $(\gamegraph, \spec)$ if every $\p{0}$ strategy following the template $\assump$ is winning from $v$. 
Moreover, we say a strategy template $\assump$ is \emph{winning} if it is winning from every vertex in $\wino$.
%
%
In addition, we call $\assump$ \emph{maximally permissive} for $\game$, if every $\p{0}$ strategy $\pi$ which is winning in $\game$ also follows $\assump$.
%
%
With slight abuse of notation, we use $ \assump $ for the set of formulas $ \set{\assump_1,\ldots, \assump_k} $, and the formula $ \assump_1\wedge\ldots\wedge\assump_k $, interchangeably.



%
 
\smallskip
\noindent\textbf{Set Transformers.} 
Let $ \gamegraph=(V=\vertexz\cupdot \vertexo, E) $ be a game graph, $ U\subseteq V $ be a subset of vertices, and $ a\in \{0,1\} $ be the player index. Then 
 \begin{align}
 	\upre{\gamegraph}{U} =& \{v\in V\mid\forall (v,u)\in E.~u\in U  \}\\
 	\cprea{\gamegraph}{U} =& \{v\in V^a\mid \exists (v,u)\in E.~u\in U \}\cup \{v\in V^{1-a}\mid u\in \upre{\gamegraph}{U}  \}
 \end{align}
 The universal predecessor operator $ \textsf{upre}_{\gamegraph}(U) $ computes the set of vertices with all the successors in $ U $ and the controllable predecessor operator $ \textsf{cpre}^a_{\gamegraph}(U) $ the vertices from which $ \newcommand{\tpre}[2]{\textsf{tpre}\ensuremath{_{#1}(#2)}}\p{a} $ can force visiting $ U $ in \emph{exactly} \emph{one} step.
 %
 In the following, we introduce two types of attractor operators: $ \textsf{attr}^a_{\gamegraph}(U) $ that computes the set of vertices from which $ \p{a}$ can force at least a single visit to $ U $ in \emph{finitely many} steps, and the universal attractor $ \textsf{uattr}_{\gamegraph}(U) $ that computes the set of vertices from which both players are forced to visit $ U $. For the following, let $ \mathsf{pre}\in \{\mathsf{upre}, \mathsf{cpre}^a \} $
 \begin{align}
 	\mathsf{pre}_{\gamegraph}^{1}{(U)} =&~\pre{\gamegraph}{U}\cup U &
 	\mathsf{pre}_{\gamegraph}^{i}{(U)}=&~\mathsf{pre}_{\gamegraph}{(\mathsf{pre}_{\gamegraph}^{i-1}{(U)})} \cup\mathsf{pre}_{\gamegraph}^{i-1}{(U)}\\
 	\attra{\gamegraph}{U} =& \cup_{i\geq 1}\cpre{\gamegraph}{U}{a,i}&
 	\uattr{\gamegraph}{U} =& \cup_{i\geq 1}\mathsf{upre}_{\gamegraph}^{i}{(U)}
 \end{align}
 \vspace{-1.1cm}
\


%

\section{Computation of Winning Strategy Templates}\label{sec:computeWST}
Given a 2-player game $ \game $ with an objective $\spec$, the goal of this section is to compute a \emph{strategy template} that characterizes a large class of winning strategies of $ \p{0} $ from a set of vertices $ U\subseteq V $ in a local, permissive, and computationally efficient way. These templates are then utilized in \cref{sec:composition} for computational synthesis. 
In particular, this section introduces three distinct template classes --- safety templates (\cref{sec:safety}), live-group-templates (\cref{sec:live}), and co-live-templates (\cref{sec:colive}) along with algorithms for their computation via safety, Büchi, and co-Büchi games, respectively. We then turn to general parity objectives which can be thought of as a sophisticated combination of Büchi and co-Büchi games.  
We show in \cref{sec:parity} how the three introduced templates can be derived for a general parity objective by a suitable combination of the previously introduced algorithms for single templates. All presented algorithms have the same worst-case computation time as the standard algorithms solving the respective game. This shows that extracting strategy \emph{templates} instead of 'normal' strategies does not incur an additional computational cost.  We prove the soundness of the algorithms and discuss the complexities in \cref{app:templates}.

\subsection{Safety Templates}\label{sec:safety}

We start the construction of strategy templates by restricting ourselves to games with a safety objective --- i.e.,  $\game=(\gamegraph,\spec)$ with $\spec\coloneqq \square U$ for some $U\subseteq V$. A winning play in a safety game never leaves $U\subseteq V$. It is well known that such games allow capturing \emph{all} winning strategies by a simple local template which essentially only allows $\p{0}$ moves from winning vertices to other winning vertices. This is formalized in our notation as a safety template as follows.,



\begin{theorem}[{\cite[Fact 7]{bernet2002:permissiveStrategies}}]\label{prop:safety strategy}
	Let $\game=(\gamegraph,\square U)$ be a safety game with winning region $\wino$ 
	and $ \safegroup = \Set{(u,v)\in E\mid \left(u\in V^0\cap \wino\right) \wedge \left(v \notin \wino\right)}$. Then
	\begin{equation}\label{eq:safety assumption definition}
		\textstyle \assumpsafe(\safegroup) \coloneqq \square \bigwedge_{e\in \safegroup} \neg e,
	\end{equation}
	is a winning strategy template for the game $\game$ which is also maximally permissive.
\end{theorem}

It is easy to see that the computation of the safety template $\assumpsafe(\safegroup)$ reduces to computing the winning region $\wino$ in the safety game $(\gamegraph,\square U)$ and extracting $\safegroup$. We refer to the edges in $\safegroup$ as \emph{unsafe edges} and we call this algorithm computing the set $\safegroup$ as $\safeTemp(\gamegraph,U)$. Note that it runs in $ \bigO(m) $ time, where $ m = |E| $, as safety games are solvable in $ \bigO(m) $ time.

%
%
%
%


\subsection{Live-Group Templates}\label{sec:live}
As the next step, we now move to simple liveness objectives which require a particular vertex set $I\subseteq V$ to be seen infinitely often. Here, winning strategies need to stay in the winning region (as before) but in addition always eventually need to make progress towards the vertex set $I$. We capture this required progress by \emph{live-group templates} --- given a group of edges $H\subseteq E$, we require that whenever a source vertex $v$ of an edge in $H$ is seen infinitely often, an edge $e\in H$ (not necessarily starting at $v$) also needs to be taken infinitely often. This template ensures that compliant strategies always eventually make progress towards $I$, as illustrated by the following example. 

\begin{example}
Consider the game graph in \cref{fig:introExample} where we require visiting $ \{c,d\} $ infinitely often. To satisfy this objective from vertex $ a $, $ \pz $ needs to not get stuck at $ a $, and should not visit $ b $ always (since $ \po $ can force visiting $ a $ again, and stop $ \pz $ from satisfying the objective). Hence, $ \pz $ has to always eventually leave $ a $ and go to $ \{c,d\} $. This can be captured by the live-group $ \{e_{ac},e_{ad}\} $. Now if the play comes to $ a $ infinitely often, $ \pz $ will go to either $ c $ or $ d $ infinitely often, hence satisfying the objective.
\end{example}

Formally, such games are called \emph{Büchi games}, denoted by $\game=(\gamegraph=(V,E),\spec)$ with $\spec\coloneqq \square \lozenge I$, for some $ I\subseteq V $. In addition, a \emph{live-group} $\livegroupSingleN = \Set{e_j}_{j\geq 0}$ is a set of edges $e_j = (s_j,t_j)$ with source vertices $\src(\livegroupSingleN):=\Set{s_j}_{j\geq 0}$. Given a set of live-groups 
 	$\livegroup=\left\{\livegroupSingle\right\}_{i\geq 0}$  we define a live-group template as
\begin{equation}\label{equ:livegroup}
	\assumpgrlive(\livegroup) \coloneqq \bigwedge_{i\geq 0}\square\lozenge src(\livegroupSingle)\implies\square\lozenge \livegroupSingle.
\end{equation}
The live-group template says that if some vertex from the source of a live-group is visited infinitely often, then some edge from this group should be taken infinitely often by the following strategy. 

Intuitively, winning strategy templates for Büchi games consist of a safety template conjuncted with a live-group template. While the former enforces all strategies to stay within the winning region $\mathcal{W}$, the latter enforces progress w.r.t.\ the goal set $I$ within $\mathcal{W}$. Therefore, the computation of a winning strategy template for Büchi games reduces to the computation of the unsafe set $S$ to define  $\assumpsafe(\safegroup)$ in \eqref{eq:safety assumption definition} and the live-group $\livegroup$ to define $\assumpgrlive(\livegroup)$ in \eqref{equ:livegroup}. We denote by $ \buchiTemp(\gamegraph, I) $ the algorithm computing the above as detailed in \cref{alg:buchitemplate}. The algorithm uses some new notations that we define here. Here, the function $\solveBuchi$ solves a Büchi game and returns the winning region (e.g., using the standard algorithm from \cite{chatterjee2008:algorithmsforbuechigames}), $ \edges (X, Y)= \{(u,v)\in E\mid u\in X, v\in Y\} $, is the set of edges between two subsets of vertices $ X $ and $ Y $. $\gamegraph|_U \coloneqq \tup{U=U^0\cupdot U^1, E'}$ s.t.\ $U^0 \coloneqq V^0\cap U$, $U^1 \coloneqq V^1\cap U$, and $E'\coloneqq E\cap (U\times U)$ denotes the restriction of a game graph $\gamegraph \coloneqq \tup{V=V^0\cupdot V^1, E}$ to a subset of its vertices $U\subseteq V$. 
%
%
%
We have the following formal result.

\begin{restatable}{theorem}{restatebuechi}\label{thm:BuechiCont}
	 Given a \buchi game $\game=(\gamegraph, \square \lozenge I)$ for some $ I\subseteq V $, if $(\safegroup, \livegroup)=\buchiTemp(\gamegraph, I)$ then
	 $\assump=\set{\assumpsafe(\safegroup), \assumpgrlive(\livegroup)} $
	 is a winning strategy template for the game $\game$, computable in time $ \bigO(nm) $, where $ n=|V|  $ and $ m=|E| $.
\end{restatable}

	\begin{algorithm}[t]
		\caption{$ \buchiTemp(\gamegraph, I) $}\label{alg:buchitemplate}
		\begin{algorithmic}[1]
			\Require A game graph $\gamegraph$, and a subset of vertices $ I$
			\Ensure A set of unsafe edges $ \safegroup $ and a set of live-groups $ \livegroup $
			\State $\wino \gets \solveBuchi (\gamegraph, I)$; $ \safegroup\gets \safeTemp(\gamegraph,\wino) $;
			\State $\gamegraph \gets \gamegraph|_{\wino}$; $I \gets I\cap \wino$; \label{step:restrict to buchi winning region}
			\State $ \livegroup\gets \reachTemp(\gamegraph, I) $;
			\State \Return $ (\safegroup, \livegroup) $
			\Procedure{\reachTemp}{$ \gamegraph, I\subseteq V $}
				\State $ \livegroup\gets\emptyset $;
				\While{$I \neq V$}
					\State $ A\gets \uattr{\gamegraph}{I} $; 
					 $ B \gets \cprez{\gamegraph}{A} $; 
					 $ \livegroup \gets \livegroup \cup \{\edges(B, A)\}$; 
					 $ I \gets A \cup B$;
				\EndWhile
				\State \Return $ \livegroup $
			\EndProcedure
		\end{algorithmic}
	\end{algorithm}

While live-group templates capture infinitely many winning strategies in \buchi games, they are \emph{not} maximally permissive, as exemplified next.
\begin{example}
 Consider the game graph in \cref{fig:introExample} restricted to the vertex set $ \{a,b,d\} $  with the \buchi objective $ \square\lozenge d $. Our algorithm outputs the live-group template $ \assump=\assumpgrlive(\{e_{ad}\}) $. Now consider the winning strategy with memory that takes edge $ e_{da} $ from $ d $, and takes $ e_{ab} $ for play suffix $ bda $ and $ e_{ad} $ for play suffix $ aba $. This strategy does not follow the template --- the play $ (abd)^{\omega} $ is in $ \lang(\stratz) $ but not in $ \lang(\assump) $.
\end{example}

%
%
%


\subsection{Co-Live Templates}\label{sec:colive}
We now turn to yet another objective which is the dual of the one discussed before. The objective requires that eventually, only a particular subset of vertices $ I $ is seen. A winning strategy for this objective would try to restrict staying or going away from $ I $ after a finite amount of time. It is easy to notice that live-group templates can not ensure this, but it can be captured by \emph{co-live templates}: given a set of edges, eventually these edges are not taken anymore. Intuitively, these are the edges that take or keep a play away from $ I $.

\begin{example}
	Consider the game graph in \cref{fig:introExample} where we require eventually stop visiting $ b $, i.e. staying in $ I=\{a,c,d\} $. To satisfy this objective from vertex $ a $, $ \pz $ needs to stop getting out of $ I $ eventually. Hence, $ \pz $ has to stop taking the edges $ \{e_{ab}, e_{db}, e_{de}\} $, which can be ensured by marking both edges co-live. Now since no edges are leading to $ b $, the play eventually stays in~$ I $, satisfying the objective. We note that this can not be captured by live-groups $ \{e_{aa},e_{ac},e_{ad}\} $ and $\{e_{da}\}$, since now the strategy that visits $ c $ and $ b $ alternatively from $\p{0}$'s vertices, does not satisfy the objective, but follows the live-group.
\end{example}

Formally, a co-B\"uchi game is a game $\game=(\gamegraph,\spec)$ with \cobuchi winning condition $\spec\coloneqq \lozenge\square I$, for some goal vertices $ I\subseteq V $. A play is winning for $ \pz $ in such a co-Büchi game if it eventually stays in $ I $ forever. The \emph{co-live} template is defined by a set of \emph{co-live} edges $ \colivegroup $ as follows,

\[ \assumpdep(\colivegroup) \coloneqq \bigwedge_{e\in\colivegroup}\lozenge\square \neg e. \]

The intuition behind the winning template is that it forces staying in the winning region using the safety template, and ensures that the play does not go away from the vertex set $ I $ infinitely often using the co-live template. We provide the procedure in \cref{alg:coBuchitemplate} and its correctness in the following theorem. Here, $ \solveCobuchi(\gamegraph, I) $ is a standard algorithm solving the \cobuchi game with the goal vertices $ I $, and outputs the winning regions for both players~\cite{chatterjee2008:algorithmsforbuechigames}. We also use the standard algorithm $ \solveSafety(\gamegraph, I) $ that solves the safety game with the objective to stay in $ A $ forever.


\begin{algorithm}[t]
	\caption{$ \coBuchiTemp(\gamegraph, I) $}\label{alg:coBuchitemplate}
	\begin{algorithmic}[1]
		\Require A game graph $\gamegraph$, and a subset of vertices $ I$
		\Ensure A set of unsafe edges $ \safegroup $ and a set of co-live edges $ \colivegroup $
		\State $ \safegroup \gets\emptyset $; $ \colivegroup\gets\emptyset $
		\State $\wino \gets \solveCobuchi (\gamegraph, I) $; $ S\gets \safeTemp(\gamegraph,\wino) $
		\State $\gamegraph \gets \gamegraph|_{\wino}$; $I \gets I\cap \wino$;\label{step:reduce to cobuchi region}
		\While{$V \neq \emptyset$}
		\State $ A\gets \solveSafety({\gamegraph},{I}) $; \label{step:colive to force staying in safety winning vertices}
		$ \colivegroup \gets \colivegroup \cup \edges(A, V\backslash A) $;
		\While{$\cprez{\gamegraph}{A} \neq A$} \label{step:inner while loop begins} \Comment{Outputs $ \textsf{attr}^0_{\gamegraph}(A) $}
		\State $ B\gets \cprez{\gamegraph}{A} $;
		\State $ \colivegroup \gets \colivegroup \cup \edges(B, V\backslash (A\cup B)) \cup \edges(B,B) $;\label{step:colive to force going toward safety winning vertices}
		\State $ A\gets A\cup B $;
		\EndWhile \label{step:inner while loop exits}
		\State $ G\gets G|_{V\backslash A} $;
		 $ I \gets I \cap V\backslash A$;
		\EndWhile
		\State 	\Return $ (\safegroup, \colivegroup) $
	\end{algorithmic}
\end{algorithm}
\begin{restatable}{theorem}{restatecobuechi}\label{thm:coBuechiCont}
	Given a \cobuchi game $\game=(\gamegraph,\lozenge \square I)$ for some $ I\subseteq V $, if $ (\safegroup,\colivegroup)=\coBuchiTemp(\gamegraph,I) $ then
	$\assump=\set{\assumpsafe(\safegroup), \assumpdep(\colivegroup)} $
	is a winning strategy template for $ \pz $, computable in time $ \bigO(nm) $ with $ n=|V|  $ and $ m=|E| $.
	
\end{restatable}


\subsection{Parity Games}\label{sec:parity}

We now consider a more complex but canonical class of $ \omega $-regular objectives. Parity objectives are of central importance in the study of synthesis problems as they are general enough to model a huge class of qualitative requirements of cyber-physical systems, while enjoying the properties like positional determinacy. 

A parity game is a game $\game=(\gamegraph,\spec)$ with parity winning condition $\spec = \paritygame(\priority)$, where

\begin{equation}
 \textstyle\paritygame(\priority)\coloneqq \bigwedge_{i\inodd [0;k]} \left(\square\lozenge P_i \implies \bigvee_{j\ineven [i+1;k]} \square\lozenge P_j\right), 
\end{equation}
%
with $ P_i=\{q\in Q\mid \priority(q)=i \} $ for some priority function $ \priority: V\rightarrow [0;d] $ that assigns each vertex a priority. A play is winning for $ \pz $ in such a game if the maximum of priorities seen infinitely often is even. 

Although parity objectives subsume previously described objectives, we can construct strategy templates for parity games using the combinations of previously defined templates. To this end, we give the following algorithm. 

\begin{algorithm}[t]
	\caption{$ \parityTemp(\gamegraph,\priority) $}
	\label{alg:parityTemplateZeilonka}
	\begin{algorithmic}[1]
		\Require A game graph $\gamegraph$, and a priority function $\priority:V\rightarrow \{0,\ldots, d\}$
		\Ensure Winning regions $(\wino,\winl),$ live-groups $\livegroup,$ and co-live edges $\colivegroup$
		\If {$ d $ is odd}
			\State $ A=\attro{\gamegraph}{P_d} $
			\If {$ A=V $}
				 \Return $ (\emptyset, V), \emptyset, \emptyset $
			\Else
				\State $ (\wino,\winl),\livegroup,\colivegroup \gets \parityTemp (\gamegraph|_{V\backslash A},\priority) $
				\If{$ \wino=\emptyset $}
					 \Return $ (\emptyset, V), \emptyset, \emptyset $
				\Else
					\State $ B = \attrz{\gamegraph}{\wino} $
					\State $ \colivegroup \gets \colivegroup\cup \edges(\wino, V\backslash \wino) $ \label{alg:step:noleavewo}
					\State $ \livegroup\gets \livegroup \cup \reachTemp(\gamegraph,\wino) $ \label{alg:step:reachwo}
					\State $ (\wino',\winl'),\livegroup',\colivegroup' \gets \parityTemp (\gamegraph|_{V\backslash B},\priority) $
					\State \Return $ (\wino'\cup B,\winl'),\livegroup\cup\livegroup',\colivegroup\cup\colivegroup'$
				\EndIf
			\EndIf
		\Else \Comment{If $ d $ is even}
			\State $ A=\attrz{\gamegraph}{P_d} $\label{step:attractevend}
			\If{A=V}\label{step:checkwholeattracts}
				 \Return  $ (V,\emptyset), \reachTemp(\gamegraph, P_d), \emptyset $\label{step:reachevend}
			\Else
				\State $ (\wino,\winl),\livegroup,\colivegroup \gets \parityTemp (\gamegraph|_{V\backslash A},\priority) $
				\If{$ \winl=\emptyset $}\label{step:nooddwinning}
					 \Return $ (V,\emptyset), \livegroup\cup \reachTemp(\gamegraph|_A, P_d), \colivegroup $\label{alg:step:reachhigheven}
				\Else
				\State $ B = \attro{\gamegraph}{\winl} $
				\State $ (\wino',\winl'),\livegroup',\colivegroup' \gets \parityTemp (\gamegraph|_{V\backslash B},\priority) $
				\State \Return $ (\wino',\winl'\cup B),\livegroup',\colivegroup'$
				\EndIf
				\EndIf
			\EndIf
	\end{algorithmic}
\end{algorithm}

\begin{restatable}{theorem}{restateparity}\label{thm:ParityCont}
	Given a parity game $\game=(\gamegraph,\paritygame(\priority))$ with priority function $ \priority:V\rightarrow [0;d] $, if $((\wino, \winl), \livegroup,\colivegroup)=\parityTemp(\gamegraph, \priority)$, then
	$ \assump=\set{\assumpsafe(\safegroup), \assumpgrlive(\livegroup), \assumpdep(\colivegroup)} $
	is a winning strategy template for the game $\game$, where $ \safegroup = \edges(\wino,\winl) $. Moreover, the algorithm terminates in time $ \bigO(n^{d+\bigO(1)}) $, which is same as that of Zielonka's algorithm.
\end{restatable}

\begin{figure}[b]
	\vspace{-0.5cm}
	\centering
	\begin{tikzpicture}
		
		\node[player0, label={below:$p_1$}] (1) at (0, 0) {$a$};
		\node[player0, label={below:$p_4$}] (2) at (\fpeval{1*\pos}, 0) {$b$};
		\node[player1, label={below:$p_5$}] (3) at (\fpeval{2*\pos},0) {$c$};
		\node[player0, label={below:$p_6$}] (4) at (\fpeval{3*\pos},0) {$d$};
		\node[player1, label={below:$p_2$}] (5) at (\fpeval{4*\pos},0) {$e$};
		\node[player0, label={below:$p_2$}] (6) at (\fpeval{5*\pos},0) {$f$};
		\node[player0, label={below:$p_1$}] (7) at (\fpeval{6*\pos},0) {$g$};
		\node[player0, label={below:$p_3$}] (8) at (\fpeval{3.5*\pos},\pos) {$h$};
		
		\path[->] (1) edge[loop above] () edge[blue!60,bend left = 20,dashed] (2);  		
		\path[->] (2) edge[bend left=20] (1) edge[bend left = 20, red!60, dotted] (3);  		
		\path[->] (3) edge[bend left=20] (2);
		\path[->] (4) edge (3);
		\path[->] (5) edge (4) edge (6);
		\path[->] (6) edge[loop above, blue!60,dashed] ();
		\path[->] (7) edge[loop above, blue!60,dashed] () edge (6);
		\path[->] (8) edge[loop left] () edge[blue!60,dashed] (4) edge (5);
	\end{tikzpicture}
	\vspace{-0.6cm}
	\caption{A parity game, where a vertex with priority $ i $ has label $ p_i $. The dotted edge in {\color{red!60} red} is a co-live edge, while the dashed edges in {\color{blue!60}blue} are singleton live-groups.}\label{fig:parityExample}
	\vspace{-0.5cm}
\end{figure}
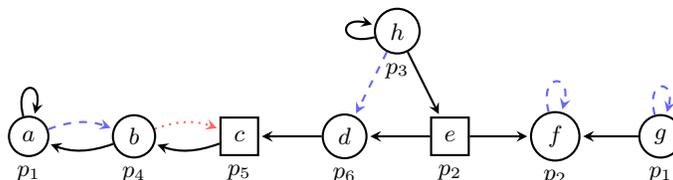

We again postpone the proof to the Appendix in \cref{app:parity}, but provide the intuition behind the algorithm and the computation of the algorithm on the parity game in \cref{fig:parityExample}. The algorithm follows the divide-and-conquer approach of Zeilonka's algorithm. Since the highest priority occurring is $ 6 $ which is even, we first find the vertices $ A=\{d,h\} $ from which $ \pz $ can force visiting $ \{d\} $ (vertices with priority 6) in \cref{step:attractevend}. Then since $ A\not=V $, we find the winning strategy template in the rest of the graph $ \gamegraph_1=\gamegraph|_{V\backslash  A} $. Then the highest priority $ 5 $ is odd, hence we compute the region $ \{c\} $ from which $ \po $ can ensure visiting $ 5 $. We again restrict our graph to $ \gamegraph_2 =G|_{\{a,b,e,f,g\}} $. Again, the highest priority is even. We further compute the region $ A_2=\{a,b\} $ from which $ \pz $ can ensure visiting the priority $ 4 $, giving us $ G_3=G|_{\{e,f,g\}} $. In $ G_3 $, $ \pz $ can ensure visiting the highest priority $ 2 $, hence satisfying the condition in \cref{step:checkwholeattracts}. Then since in this small graph, $ \pz $ needs to keep visiting priority $ 2 $ infinitely often, which gives us the live-groups $ \{e_{gf}\} $ and $ \{e_{ff}\} $ in \cref{step:reachevend}. Coming one recursive step back to $ G_2 $, since $ G_3 $ doesn't have a winning vertex for $ \po $, the if condition in the \cref{step:nooddwinning} is satisfied. Hence, for the vertices in $ A_2 $, it suffices to keep visiting priority $ 4 $ to win, which is ensured by the live-group $ \{e_{ab}\} $ added in the \cref{alg:step:reachhigheven}. Now, again going one recursive step back to $ G_1 $, we have $ \wino=\{a,b,e,f,g\} $. If $ \pz $ can ensure reaching and staying in $ \wino $ from the rest of the graph $ G_1 $, it can satisfy the parity condition. Since from the vertex~$ c $, $ \wino $ will anyway be reached, we get a co-live edge $ e_{bc} $ in \cref{alg:step:noleavewo} to eventually keep the play in $ \wino $. Coming back to the initial recursive call, since now again $ G_1 $ was winning for $ \pz $, they only need to be able to visit the priority $ 6 $ from every vertex in $ A $, giving another live-group $ \{e_{hd}\} $.


\section{Extracting Strategies from Strategy Templates}\label{sec:implementability}
This section discusses how a strategy that follows a computed winning strategy template can be extracted from the template. As our templates are just particular LTL formulas, one can of course use automata-theoretic techniques for this.
However, as the types of templates we presented put some local restrictions on strategies, we can extract a strategy much more efficiently.
For instance, the game in \cref{fig:introExample} with strategy template $\assump= \assumpgrlive(\{e_{ac},e_{ad}\})$ 
allows the strategy that simply uses the edges $e_{ac}$ and $e_{ad}$ alternatively from vertex $a$.

However, strategy extraction is not as straightforward for every template, even if it only conjuncts the three template types we introduced in \cref{sec:computeWST}. For instance, consider again the game graph from \cref{fig:introExample} with a strategy template \linebreak $\assump =\set{\assumpsafe(e_{ac}, e_{ad}), \assumpdep(e_{aa}, e_{ab})}$. Here, non of the four choices of $\p{0}$ (i.e., outgoing edges) from vertex $a$ can be taken infinitely often, and, hence, the only way a play satisfies $\assump$ is to not visit vertex~$a$ infinitely often. 
%
On the other hand, given strategy template $\assump' =\{ \assumpdep(e_{ab}, e_{db}), \assumpgrlive(\{e_{ab},e_{ac},e_{db}\})\}$, edge $e_{db}$ is both live and co-live, which raises a conflict for vertex $d$. Hence, the only way a strategy can follow $\assump'$ is again to ensure that $d$ is not visited infinitely often. 
%
We call such situations \emph{conflicts}. 
Interestingly, the methods we presented in \cref{sec:computeWST} never create such conflicts and the computed templates are therefore \emph{conflict-free}, as formalized next and proven in \cref{app:prop:implementability}. 
\begin{definition}\label{def:implementability}
A strategy template $\assump = \set{\assumpsafe(\safegroup), \assumpdep(\colivegroup), \assumpgrlive(\livegroup)}$ in a game graph $\gamegraph = (V,E)$ is \emph{conflict-free} if the following are true:
\begin{compactenum}[(i)]
\item or every vertex $v$, there is an outgoing edge that is neither co-live nor unsafe, i.e., 
$v\times E(v) \not \subseteq D \cup S$, and
\item for every source vertex $v$ in a live-group $\livegroupSingleN\in \livegroup$, there exists an outgoing edge in $\livegroupSingleN$ which is neither co-live nor unsafe, i.e.,
$v\times \livegroupSingleN(v) \not \subseteq D \cup S$.
\end{compactenum}
\end{definition}

\begin{restatable}{proposition}{restateImplementability}\label{prop:implementability}
 Algorithms~\ref{alg:buchitemplate}, \ref{alg:coBuchitemplate}, and \ref{alg:parityTemplateZeilonka} 
 always return conflict-free templates.
\end{restatable}

Due to the given conflict-freeness, winning strategies are indeed easy to extract from winning strategy templates, as formalized next.

\begin{proposition}
Given a game graph $\gamegraph = (V,E)$ with \emph{conflict-free} winning strategy template
 $\assump = \set{\assumpsafe(\safegroup), \assumpdep(\colivegroup), \assumpgrlive(\livegroup)}$, a winning strategy $\pi_0$ that follows $\assump$ can be extracted in time $ \bigO(m) $, where $ m $ is the number of edges.
\end{proposition}

The proof is straightforward by constructing the winning strategy as follows.
We first remove all unsafe and co-live edges from $G$ and then construct a strategy $\pi_0$ that alternates between all remaining edges from every vertex in~$\mathcal{W}_0$. This strategy is well defined 
as condition (i) in \cref{def:implementability} ensures that after removing all the unsafe and co-live edges a choice from every vertex remains. Moreover, if the vertex is a source of a live-group edge, condition (ii) in \cref{def:implementability} ensures that  there are outgoing edges satisfying every live-group. It is easy to see that the constructed strategy indeed follows $\assump$ and is hence winning from vertices in $\mathcal{W}_0$, as $\assump$ was a winning strategy template. We call this procedure of strategy extraction $ \extract(\gamegraph, \assump) $.


\section{Applications of Strategy Templates}\label{sec:applications}
This section considers two concrete applications of strategy templates which utilize their structural simplicity and easy adaptability. 

In the context of CPS control design problems, it is well known that the game graph of the resulting parity game used for strategy synthesis typically has a physical interpretation and results from behavioral constraints on the \emph{existing technical system} that is subject to control. In particular,  following the well-established paradigm of abstraction-based control design (ABCD) \cite{tabuada2009verification,alur2015principles,belta2017formal}, an underlying (stochastic) disturbed non-linear dynamical system can be automatically abstracted into a two-player game graph using standard abstraction tools, e.g.\ SCOTS~\cite{SCOTS}, ARCS~\cite{ARCS}, MASCOT~\cite{MASCOT}, P-FACES~\cite{PFACES}, or ROCS~\cite{ROCS}. 


In contrast to classical problems in reactive synthesis, it is very natural in this context to think about the game graph and the specification as two \emph{different} objects. Here, specifications are naturally expressed via propositions that are defined over sets of states of this underlying game graph, without changing its structure. This separation is for example also present in the known LTL fragment GR(1) \cite{piterman2006synthesis}. Arguably, this feature has contributed to the success of GR(1)-based synthesis for CPS applications, e.g.\ \cite{WongpiromsarnTopcuMurray_MPC_LTL,EhlersKressGazit_robotsGR1,alur2013counter,MaozRingert2015,KressGazit2007,KressGazit2009,SVORENOVA2017}.  

Given this insight, it is natural to define the incremental synthesis problem such that the game graph stays unchanged, while newly arriving specifications are modeled as new parity conditions over the same game graph. Formally, this results in a \emph{generalized parity game} where the different objectives arrive \emph{one at a time}. We show an incremental algorithm for synthesizing winning strategies for such games in \cref{sec:composition}. 
Similarly, fault-tolerant control requires the controller to adapt to unavailable actuators within the technical system under control. This naturally translates to the removal of $\pz$ edges within the game graph given its physical interpretation. 
We show how strategy templates can be used to adapt winning strategies to these game graph modifications in \cref{sec:fault}. 

\subsection{Incremental Synthesis via Strategy Templates}\label{sec:composition}
In this section we consider a $2$-player game $\game$ with a conjunction $\spec=\bigwedge_{i=1}^k\spec_i$ of multiple parity objectives $\spec_i$, also called a \emph{generalized} parity objective. However, in comparison to existing work \cite{Buryere2019:partialsolvers,chatterjee2007:generalizedparitygames}, we consider the case that different objectives~$\spec_i$ might not arrive all at the same time. The intuition of our algorithm is to solve each parity game $(G,\spec_i)$ separately and then combine the resulting strategy templates $\assump_i$ to a global template $\assump=\bigwedge_{i=1}^k\assump_i$. This allows to easily incorporate newly arriving objectives $\spec_{k+1}$. We only need to solve the parity game $(G,\spec_{k+1})$ and then combine the resulting template $\assump_{k+1}$ with $\assump$.  

While \cref{prop:implementability} ensures that every individual template $\assump_i$ is \emph{conflict-free}, this does unfortunately not imply that their conjunction is also \emph{conflict-free}. Intuitively, combinations of strategy templates can cause the condition (i) and (ii) in \cref{def:implementability} to not hold anymore, resulting in a \emph{conflict}.
%
As already discussed in \cref{sec:implementability}, this requires source vertices $U\subseteq V$ with such conflicts to eventually not be visited anymore. We therefore resolve such conflicts by adding the specification $\lozenge\square\neg U$ to every objective and recomputing the templates.

To efficiently formalize this objective change, we note that a parity objective $\paritygame(\priority)$ with an additional specification $\lozenge\square \neg U$ for some $U\subseteq V$ is equivalent to another parity objective $Parity(\priority')$, where priority function $\priority'$ can be obtained from $\priority:V\rightarrow[0;2d+1]$ just by modifying the priorities of vertices in $U$ to $2d+1$. Let us denote such a priority function by $\priority[U\rightarrow 2d+1]$. In particular, we have the following result:

\begin{restatable}{lemma}{reductionToSingleParity}\label{lem:compositon}
Given a game graph $\gamegraph$ and two parity objectives $\spec = \paritygame(\priority)$, $\spec' = Parity(\priority')$ such that $\priority:V\rightarrow [0;2d+1]$ and $\priority' = \priority[U\rightarrow 2d+1]$ for some vertex set $U\subseteq V$, it holds that $\lang(\spec') = \lang(\spec \wedge \lozenge\square\neg U)$. Moreover, if a strategy template is winning from some vertex $u$ in the game $\game' = (\gamegraph, \spec')$, then it is also winning from $u$ in the game $\game = (\gamegraph, \spec)$.
\end{restatable}


\begin{algorithm}[t]
	\caption{$\compose(\gamegraph , (\wino', \livegroup', \colivegroup', (\spec_i)_{i < \ell}),(\spec_i)_{\ell\leq i\leq k})$ where $\spec_i = \paritygame(\priority_i)$}\label{alg:composition}
	\begin{algorithmic}[1]
		\Require A generalized parity game $\gamegraph = (V,E)$ and  objectives $(\spec_i)_{i\leq k}$ with $\spec_i = \paritygame(\priority_i)$ such that $\priority_i : V \rightarrow [0;2d_i+1]$ along with a partial winning region, live-groups, and co-live edges $(\wino,\livegroup,\colivegroup)$ for the generalized parity game $(\gamegraph, \bigwedge_{i<\ell}\spec_i)$.
		\Ensure A partial winning region $\wino$, live-groups $\livegroup$, co-live edges $\colivegroup$, and modified parity objectives  $(\spec'_i)_{i\leq k}$.
		\State $(W_i, V\setminus W_i), \livegroup_i, \colivegroup_i \gets \parityTemp(\gamegraph|_{\wino},\spec_i)$ for each $\ell \leq i\leq k$\label{alg:composition:parity}
		\State $\livegroup =\livegroup' \cup \bigcup_{\ell\leq i\leq k}\livegroup_i$; $\colivegroup = \colivegroup' \cup \bigcup_{\ell\leq i\leq k}\colivegroup_i$;  $\wino = \wino' \cap \bigcap_{\ell\leq i\leq k}  W_i$
		\State $\conflict_1 = \{u\in \wino \mid u \times (E(u)\cap \wino) \subseteq \colivegroup\}$
		\State $\conflict_2 =\{u\in \wino \mid u \times (\livegroupSingleN(u)\cap \wino) \subseteq \colivegroup, ~\livegroupSingleN \in  \livegroup,~\livegroupSingleN(u) \neq \emptyset\}$
		\If{$\conflict_1 \cup \conflict_2  = \emptyset$}\label{alg:composition:if}
			\State 	\Return $ (\wino, \livegroup, \colivegroup, (\spec_i)_{i\leq k}) $
		\Else
			\State $\priority'_i(u) \gets \priority[\conflict_1\cup \conflict_2 \rightarrow 2d_i'+1]$ for each $i\leq k$ \label{alg:compositon:line8}
			\State 	\Return $\compose(\gamegraph, (\wino, \emptyset, \emptyset, \emptyset), (\spec'_i)_{i\leq k})$ with $\spec'_i = \paritygame(\priority'_i))$
		\EndIf
	\end{algorithmic}
\end{algorithm}

Using the above ideas,  we present \cref{alg:composition} to solve generalized parity games (possibly incrementally).
If no partial solution to the synthesis problem exists so far we have $\ell=0$, otherwise the game $(\gamegraph, \bigwedge_{i<\ell}\spec_i)$ was already solved and the respective winning region and templates are known. 
In both cases, the algorithm starts with computing a winning strategy template for each game $(\gamegraph,\spec_i)$ for $i\in\{\ell+1,k\}$ (line 1) and conjuncts them with the already computed ones (line 2). Then the algorithm checks for conflicts (line 3-4).
If there is some conflict the algorithm modifies the objectives to ensure that the conflicted vertices are eventually not visited anymore (line 8), and then re-computes the templates in the game graph restricted to the intersection of winning regions for all objectives (line 9). 
If there is no conflict,  then the algorithm returns the conjunction of the templates which is conflict-free, and hence, is winning from the intersection of winning regions for every objective (line 6). The latter is formalized in the following theorem. The proof can be found in \cref{app:correctnessOfCompositionality}.

\begin{restatable}{theorem}{correctnessOfCompositionality}\label{thm:compositionCorrectness}
Given a generalized parity game $\game = (\gamegraph, \bigwedge_{i\leq k}\spec_i)$ with $\spec_i = \paritygame(\priority_i)$ and  priority functions $\priority_i: V \rightarrow [0;2d_i+1]$, if
$(\wino, \livegroup, \colivegroup, (\spec'_i)_{i\leq k}) = \compose(\gamegraph, \emptyset, (V, \emptyset, \emptyset), (\spec_i)_{i\leq k})$, then
	$\assump=\{\assumpsafe(\safegroup),\assumpgrlive(\livegroup),$\linebreak $ \assumpdep(\colivegroup)\} $
	is an conflict-free strategy template that is winning from $\wino$ in the game $\game$, where $ \safegroup = \edges(\wino,V\setminus\wino). $
Further, $\assump$ is computable in time $\bigO(kn^{2d+3})$ time, where $n = \abs{V}$ and $d = \max_{i\leq k} d_i$.
\end{restatable}

Due to the conflict checks carried out within \cref{alg:composition} the returned modified objectives $\spec'_i$ ensure that the \emph{conjunction} $\assump:=\bigwedge_{i=1}^k\assump'_i$ of winning strategy templates $\assump'_i$ for the games $(\gamegraph,\spec'_i)$ is indeed conflict-free. In particular, the conjuncted template $\assump$ is actually returned by the algorithm. Hence, incrementally running \cref{alg:composition} is actually sound. This is an immediate consequence of \cref{thm:compositionCorrectness} and stated as a corollary next.
%
\begin{corollary}
Given a generalized parity game $\game = (\gamegraph, \bigwedge_{i\leq k}\spec_i)$ with $\spec_i = \paritygame(\priority_i)$ and  priority functions $\priority_i: V \rightarrow [0;2d_i+1]$, s.t.\
\begin{align*}
(\wino', \livegroup', \colivegroup', (\spec'_i)_{i<\ell}) &:= \compose(\gamegraph, (V, \emptyset, \emptyset, \emptyset), (\spec_i)_{i<\ell}), \text{ and}\\
(\wino, \livegroup, \colivegroup, (\spec''_i)_{i\leq k}) &:= \compose(\gamegraph, (\wino', \livegroup', \colivegroup', (\spec'_i)_{i<\ell}), (\spec_i)_{\ell \leq i\leq k})
\end{align*}
then
	$\assump=\set{\assumpsafe(\safegroup), \assumpgrlive(\livegroup), \assumpdep(\colivegroup)} $
	is an conflict-free strategy template that is winning from $\wino$ in the game $\game$, where $ \safegroup = \edges(\wino,V\setminus\wino). $
Further, $\assump$ is computable in time $\bigO(kn^{2d+3})$, where $n = \abs{V}$ and $d = \max_{i\leq k} d_i$.
\end{corollary}

We note that the generalized Zielonka algorithm~\cite{chatterjee2007:generalizedparitygames} for solving generalized parity games has time complexity $\bigO(mn^{\sum 2d_i})  {\sum{d_i} \choose d_1,d_2,\ldots,d_k}$ for a game with $n$ vertices, $m$ edges and $k$ priority functions: $\priority_i$ with $2d_i$ priorities for each $i$. Clearly, \cref{alg:composition} has a much better time complexity. However, it is not complete, i.e, it does not always return the complete winning region. 
This is due to templates being not maximally permissive and hence potentially raising conflicts which result in additional specifications that are not actually required. 
The next example shows such an incomplete instance for illustration. We however note that \cref{alg:composition} returned the \emph{full winning region} on \emph{all} benchmarks considered during evaluation, suggesting that such instances rarely occur in practice.

\begin{example}
 Consider the game in \cref{fig:introExample} with objectives $\spec_3 \wedge \spec_4$ with $\spec_4 = \paritygame(\priority)$, where $\priority$ maps vertices $a,b,c,d,e,f$ to $0,2,1,1,1,1$, respectively. 
 The winning strategy templates computed by $\parityTemp$ for objectives $\spec_3$ and $\spec_4$ are $\assump_3 = \assumpdep(e_{ab},e_{db},e_{de})$ and $\assump_4 = \assumpgrlive(\{e_{ab},e_{db},e_{de}\})$, respectively. 
The conjunction of both templates marks all outgoing edges of vertex $a$ and $d$ in the live-group co-live. Hence, the algorithm would ensure that these conflicted vertices $a$ and $d$ are eventually not visited anymore. However, the only way to satisfy $\spec_3\wedge\spec_4$ is by eventually looping on vertex~$a$. But this solution was skipped by the strategy template $\assump_4$ by putting edge $e_{ab}$ in a live-group. 
Therefore, the algorithm resturns the empty set as the winning region, whereas the actual winning region is the whole vertex set.
\end{example}



\newcommand{\fault}{\ensuremath{\mathcal{F}}}
\subsection{Fault-Tolerant Strategy Adaptation}\label{sec:fault}
 
 \begin{algorithm}[b]
	\caption{$ \PUFcorrection(\gamegraph,\assump, \faulty) $}
	\label{alg:PUFalgo}
	\begin{algorithmic}[1]
		\Require A parity game $\game = (\gamegraph, \paritygame(\priority))$, a strategy template $ \assump$, and a set of faulty edges $ \faulty $
		\Ensure A new strategy template $ \assump' $
		\State $ \assump' \gets \set{\assump,\assumpsafe(\faulty)}$
		\If {\checkconflictfree($\gamegraph$,$\assump'$)} 
		\Return $ \assump'$
		\Else
		\State \Return $\parityTemp(\gamegraph|_{E\backslash \faulty}, \priority|_{E\backslash \faulty}) $
		\EndIf
	\end{algorithmic}
\end{algorithm}

In this section we consider a $2$-player parity game $ \game=(\gamegraph,\paritygame(\priority)) $ and a set of faulty $\pz$ edges $ \faulty\subseteq E\cap (V^0\times V) $ which might become unavailable during runtime. Given a strategy template $\assump$ for $\game$, we can use $\assump' = \set{\assump, \assumpsafe(\faulty)}$ for the (linear-time) extraction of a new strategy for the game, if $\assump'$ is conflict-free for $\gamegraph$. In this case, no re-computation is needed. If $\assump'$ is not conflict-free for $\gamegraph$, then we can remove the edges in $\faulty $ and compute a new winning strategy template using \cref{alg:parityTemplateZeilonka}. This is formalized in \cref{alg:PUFalgo}, where we slightly abuse notation and assume that $ \parityTemp $ only outputs strategy templates. The correctness of \cref{alg:PUFalgo} follows directly from \cref{thm:ParityCont}. 
\begin{corollary}
 Given a $2$-player parity game $ \game=(\gamegraph,\paritygame(\priority)) $ with a strategy template $\assump= \parityTemp(\gamegraph, \priority) $ 
 and faulty edge set $ \faulty\subseteq E\cap (V^0\times V) $ it holds that $\assump'$ obtained from \cref{alg:PUFalgo} is a winning strategy template for $ \game|_{E\setminus \faulty} $.
\end{corollary}
Faulty edges introduce an additional safety specification for which our templates are maximally permissive. This implies that \cref{alg:PUFalgo} is \emph{sound and complete} -- if there exists a winning strategy for $(G|_{E\setminus F},\paritygame(\priority))$ \cref{alg:PUFalgo} finds one. 

Let us now assume that $\faulty$ collects all edges controlling \emph{vulnerable} actuators that \emph{might} become unavailable. In this scenario, \cref{alg:PUFalgo} returns a conservative strategy that \emph{never} uses vulnerable actuators. It might however be desirable to use actuators as long as they are available to obtain better performance.
Formally, this application scenario can be defined via a time-dependent graph who's edges change over time, i.e., $ E_t $ with $ E_0=E $ are the edges available at time $t\in\mathbb{N}$ and $ \faulty:=\{e\in E\mid e\not \in E_i, \text{ for some } i\} $.
Given the original parity game $ \game=(\gamegraph,\paritygame(\priority)) $ with a winning strategy template $\assump$ we can easily modify $\extract(\game,\assump)$ to obtain a time-dependent strategy $ \strat_g $ which reacts to the unavailability of edges, i.e., at time $ t $, $ \strat_g $ takes an edge $ e\in E_t\backslash(\safegroup\cup\colivegroup) $ for all vertices without any live-group, and for the ones with live-groups, it alternates between the edges satisfying the live-groups whenever they are available, and an edge $ e\in E_t\backslash(\safegroup\cup\colivegroup) $ when no live-group edge is available. 

The online strategy $\strat_g$ can be implemented even without knowing when edges are available\footnote{We note that it is reasonable to assume that current actuator faults are visible to the controller at runtime, see e.g.\ \cite{REIJNEN2021103473} for a real water gate control example.}, i.e., without knowing the time dependent edge sequence $\set{E_t}_{t\in\mathbb{N}}$ up front. In this case $ \strat_g $ is obviously winning in  $ \game=(\gamegraph,\paritygame(\priority)) $ if $\assump$ is conflict-free for $\game|_{E\setminus F}$. 
If this is not the case, one needs to ensure that edges that cause conflicts are always eventually available again, as formalized next.
%

\begin{definition}
Given a parity game $ \game=(\gamegraph,\paritygame(\priority)) $ we call the dynamic edge set $\set{E_i}_{i\geq 0}$ a \emph{guaranteed availability fault (GAF)} if $ \forall $ plays $ \play=v_0v_1\ldots $, $ \forall v\in V $, if $ v\in inf(\play) $, then $ \forall e=(v,w)\in \faulty $, $\exists $ infinitely many times $ t_0,t_1\ldots $ such that $ v_{t_j}=v $ and $ e\in E_{t_j} $, $ \forall j\geq 0 $.
\end{definition}

Intuitively, guaranteed availability faults (GAF) ensure that a faulty edge is always eventually available when a play is in its source vertex. Under this fault, the following fault-correction result holds, which is proven in \cref{sec:gafcorrection}.

\begin{restatable}{proposition}{GAFcorrection}\label{prop:gafcorrection}
	Given a game graph $\gamegraph $ with a parity objective $ \spec $, a strategy template $\assump= \set{\assumpsafe(\safegroup), \assumpgrlive(\livegroup), \assumpdep(\colivegroup)} $ computed by \cref{alg:parityTemplateZeilonka} and a set $ \faulty = \{e\in E\mid e\not \in E_i, \text{ for some } i\} $ of faulty edges, the game with the objective is realizable under GAF if for every vertex $ v\in V^0 $, there is an outgoing edge which is not in $ \safegroup\cup\colivegroup\cup \faulty $.
\end{restatable}

This proposition allows a simple linear-time algorithm to check if the templates computed by \cref{alg:parityTemplateZeilonka} are GAF-tolerant: check if every vertex in the winning region has an outgoing edge which is not in $ \safegroup\cup\colivegroup\cup \faulty $. If this is not the case, the recomputation is non-trivial and is out of scope of this paper. We can however collect the vertices which do not satisfy the above property and alert the system engineer that these vulnerable actuators require additional maintenance or protective hardware.
%
Our experimental results in \cref{sec:experiments} show that conflicts arising from actuator faults are rare and very local. Our strategy templates allow to easily localize them, which supports their use for CPS applications.

\section{Empirical Evaluation}\label{sec:experiments}

We have developed a C++-based prototype tool~\toolname\footnote{Repository URL: \url{https://github.com/satya2009rta/pestel}} (computing \textbf{\textsc{Pe}}rmissive \textbf{\textsc{S}}trategy \textbf{\textsc{Te}}mp\textbf{\textsc{l}}ates) that implements algorithms \ref{alg:buchitemplate} -- \ref{alg:PUFalgo}. We have used \toolname to show its superior performance on the two applications considered in \cref{sec:applications}, suggesting its practical relevance. All our experiments were performed on a computer equipped with Apple M1 Pro 8-core CPU and 16GB RAM.

\smallskip
\noindent\textbf{Incremental Synthesis.} We used \toolname to solve generalized parity games both in one shot and incremental. 
We compare our algorithm with existing algorithms, i.e., \genziel from \cite{chatterjee2007:generalizedparitygames} and three partial solvers\footnote{While \genziel is sound and complete \cite{chatterjee2007:generalizedparitygames}, we found different randomly generated games where the algorithms from \cite{Buryere2019:partialsolvers} either return a superset or a subset of the winning region, hence compromising soundness and completeness. Since \cite{Buryere2019:partialsolvers} lacks rigorous proof, it is not clear whether this is an implementation bug or a theoretical mishap, leaving soundness and completeness guarantees of these algorithms open.} from \cite{Buryere2019:partialsolvers}, by executing them on a large set of benchmarks.
We have generated two types of benchmarks from the games used for the Reactive Synthesis Competition (SYNTCOMP)~\cite{benchmark:syntcomp}. 
Benchmark A was generated by converting parity games into Streett games using standard methods, and as each Streett pair can be represented by a $\{0,1,2\}$-priority parity game, we represented the complete Streett objective as a conjunction of multiple $\{0,1,2\}$-priority parity objectives, resulting in a generalized parity game.
Benchmark B was generated by adding randomly\footnote{The random generator takes three parameters: game graph “G”, number of objectives "k”, and maximum priority “m”; and then it generates “k” random parity objectives with maximum priority “m” as follows: 50\% of the vertices in “G” are selected randomly, and those vertices are assigned priorities ranging from 0 to “m” (including 0 and m) such that 1/m-th (of those 50\%) vertices are assigned priority 0 and 1/m-th are assigned priority 1 and so on. The rest 50\% are assigned random priorities ranging from 0 to “m”. Hence, for every priority, there are at least 1/(2m)-th vertices (i.e., 1/m-th of 50\% vertices) with that priority.} generated parity objectives to given parity games. 
We considered $200$ examples in Benchmark A and more than $1400$ examples in Benchmark B.

\begin{table}[t]
 \scriptsize
 \centering
 \begin{tabular}{|c|c|c|c|c|c|c|}
 \hline
                      &             & \toolname & \genziel \cite{chatterjee2007:generalizedparitygames} & \begin{tabular}[c]{@{}c@{}}\genziel\&\\ Gen\buchi\cite{Buryere2019:partialsolvers}\end{tabular} & \begin{tabular}[c]{@{}c@{}}\genziel\&\\ GenGoodEp\cite{Buryere2019:partialsolvers}\end{tabular} & \begin{tabular}[c]{@{}c@{}}\genziel\&\\ GenLay\cite{Buryere2019:partialsolvers}\end{tabular} \\ \hline
 \multirow{4}{*}{\shortstack{Benchmark A\\ (one shot)}}   & mean time   & 343 & 64 & 68 & 553  & 1224  \\ \cline{2-7} 
                      								   & incomplete  & 0     & -    & 3    & 3      & 2         \\ \cline{2-7} 
                                                       &{\cellcolor{lightgray!50} faster than } & {\cellcolor{lightgray!50}-}   & {\cellcolor{lightgray!50}74\% }   & {\cellcolor{lightgray!50}75\% }  & {\cellcolor{lightgray!50}96\%}   & {\cellcolor{lightgray!50}85\%}    \\ \cline{2-7} 
                      								   & timeouts    & 0  & 0  & 0   & 2  & 20 \\ \hline
 \multirow{4}{*}{\shortstack{Benchmark B\\ (one shot)}} & mean time   & 60 & 47  & 58   & 112   & 171 \\ \cline{2-7} 
                      									& incomplete  & 0 & -  & 28  & 27   & 2   \\ \cline{2-7} 
                      									&{\cellcolor{lightgray!50}faster than} & {\cellcolor{lightgray!50}-}  & {\cellcolor{lightgray!50}93\% } & {\cellcolor{lightgray!50}93\% } & {\cellcolor{lightgray!50}97\%}  & {\cellcolor{lightgray!50}95\%}  \\ \cline{2-7} 
                      									& timeouts    & 1  & 0   & 2   & 4  & 18   \\ \hline
                      									 Overall										   &{\cellcolor{lightgray!50} faster than}  & {\cellcolor{lightgray!50}-}   & {\cellcolor{lightgray!50} 90\%}    & {\cellcolor{lightgray!50}90\% }  & {\cellcolor{lightgray!50}97\%}   & {\cellcolor{lightgray!50}94\% }     \\ \hline \hline 
                      									 
 \multirow{4}{*}{\shortstack{Benchmark B\\(incremental)}} & mean time   &91  & 208  & 215  & 338  & 394   \\ \cline{2-7} 
                      									 & incomplete  & 0 & -   & 24  & 23   & 2   \\ \cline{2-7} 
                      									 &{\cellcolor{lightgray!50} faster than} & {\cellcolor{lightgray!50}-}  & {\cellcolor{lightgray!50}97\%}   & {\cellcolor{lightgray!50}97\%}    & {\cellcolor{lightgray!50}98\%}   & {\cellcolor{lightgray!50}99\% } \\ \cline{2-7} 
                      									 & timeouts    & 2   & 0   & 0  & 8   & 23 \\ \hline\hline
 \end{tabular}
\vspace{0.2cm}
 \caption{Aggregated experimental results on generalize parity game benchmarks with objectives given \emph{up-front} (top) and \emph{one-by-one} (bottom). Subrows: 1st row (mean time) -- average computation time (in ms); 2nd row (incomplete) -- number of examples where the corresponding tool failed to compute the complete winning region; 3rd row (faster than) -- number of examples where \toolname is faster than the respective tool; 4th row (timeouts) -- number of examples where the respective tool timed out (10000 ms).}\vspace{-0.8cm}
  \label{table:experiment}
 \end{table}

We summarize the complete set of results of the experiments in\footnote{See \cref{app:experiments} for a version of \cref{fig:intro-experiement-sensitivity} including all solvers considered in \cref{table:experiment}.} \cref{table:experiment} and \cref{fig:intro-experiement-sensitivity}.
We performed two kinds of experiments. First, we solved every generalized parity game in Benchmark A and B in \emph{one shot} using the different methods. The results are shown in \cref{table:experiment} (top) and \cref{fig:intro-experiement-sensitivity} (left). Although the average time taken by \toolname is higher than \genziel and one partial solver, it is fastest in more than 90\% of the games in both benchmarks. Thus, it shows that \toolname is as efficient as the other methods in most cases. Moreover, for every game in both benchmarks, \toolname succeeded to compute the complete winning region, whereas the partial solvers failed to do so in some cases\footnote{Additionally, we outperform all algorithms on the benchmarks considered by Bruy\`ere et al.~\cite{Buryere2019:partialsolvers}. We have however chosen to not include them in our analysis as many of their generalized parity games have only one objective and are therefore trivial.}.. We note that the instances which are hard for \toolname are those where the winning region becomes empty, which is quickly detected by \genziel but only seen by \toolname after most objectives are (separately) considered.

Second, we solved the examples in Benchmark B by adding the objectives \emph{one-by-one}, i.e., we solved the game with one objective, then we added one more objective and solved it again, and so on.  The results are shown in \cref{table:experiment} (bottom) and \cref{fig:intro-experiement-sensitivity} (right).
As \toolname can use the pre-computed strategy templates if we add a new objective to a game, it outperforms all the other solvers significantly as they need to re-solve the game from scratch every time.  

\smallskip
\noindent\textbf{Fault-tolerant Control.} 
As discussed in \cref{sec:fault}, strategy templates can be used to implement a fault tolerant time-dependent strategy, if the set of faulty edges $F$ does not cause conflicts with the strategy template. We have used \toolname on over $200$ examples of parity games from SYNTCOMP~\cite{benchmark:syntcomp} to evaluate the relevance of such conflicts in practice. For this, we randomly selected different percentages of edges to be faulty and checked for conflicts with the given template. The results are summarized in \cref{fig:experiment-fault}.
The left plot shows the number of instances for which a conflict occurs if a certain percentage of randomly selected edges is faulty. We see that the majority of the instances never faces a conflict even when $30\%$ of the edges are faulty. Looking more closely into the instances with conflicts, \cref{fig:experiment-fault} (right) shows the average number of conflicting vertices in these benchmarks. Here we see that conflicts occur very locally at a very small number of vertices. Our strategy templates allow for a linear-time algorithm to localize them, allowing to mitigate them in practice by additional hardware.


 \begin{figure}[t]
 \begin{center}
     \includegraphics[scale=0.28]{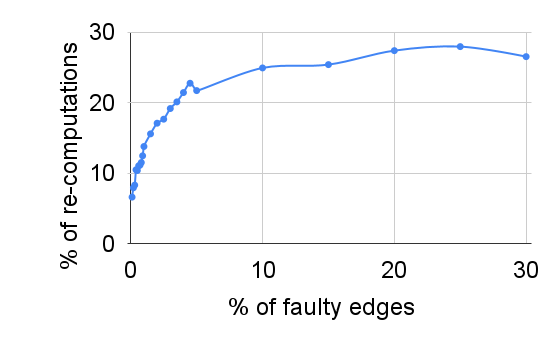}
     \includegraphics[scale=0.28]{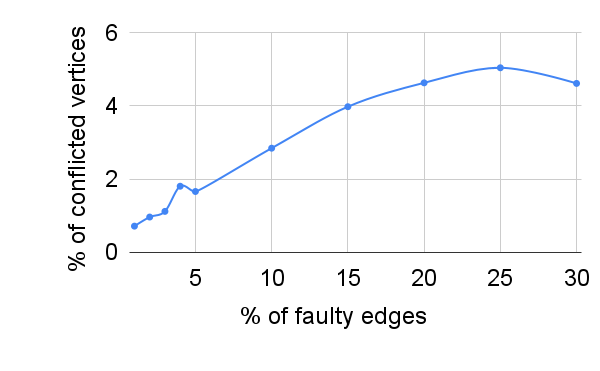}
     \end{center}
     \vspace{-0.9cm}
      \caption{Experimental results for parity games with faulty edges. Left: percentage of instances with conflicts given a certain percentage of faulty edges. Right: average percentage of vertices that created conflicts given a certain percentage of faulty edges.  }
      \label{fig:experiment-fault}
      \vspace{-0.5cm}
  \end{figure}

  \begin{remark}
   We remark again that our results are directly applicable to CPS with continuous dynamics via the paradigm of abstraction-based control design (ABCD).  In particular, standard abstraction tools such as SCOTS~\cite{SCOTS}, ARCS~\cite{ARCS}, MASCOT~\cite{MASCOT}, P-FACES~\cite{MASCOT}, or ROCS~\cite{ROCS} automatically compute a game graph from the (stochastic) continuous dynamics that can directly be used as an input to \toolname. The winning strategy computed by \toolname can further be refined into a correct-by-construction continuous feedback controller for the original dynamical system using standard methods from ABCD. We leave these tool integrations to future work.
  \end{remark}

%
%
%
\bibliographystyle{splncs04}
\newpage
\bibliography{main}

\newpage
\appendix

 \section{Winning Strategy Templates}\label{app:templates}

\subsection{Strategy Templates for \buchi Games}\label{app:Buechi}

Here, we restate \cref{thm:BuechiCont}, and formally prove the same. 
\restatebuechi*

\begin{proof}
	Before proceeding with the proof of \cref{thm:BuechiCont} we show that \cref{alg:buchitemplate} terminates.
\begin{lemma}\label{lemma:buchi algo terminates}
	The \cref{alg:buchitemplate} terminates in time $ \bigO(nm) $, where $ n $ and $ m $ are as above.
\end{lemma}
\begin{claimproof}
	Let $ k\in \N $ be such that $ I=I_0\subseteq I_1\subseteq \cdots I_k=I_{k+1} $, where $ I_j $ is the value of $ I $ in the $ j $-th iteration of the while loop in the \reachTemp ~procedure. It suffices to show that $ I_k=V $, for the graphs where $ V=\solveBuchi(\gamegraph,I) $, since we have already restricted our initial graph to such a graph in \cref{step:restrict to buchi winning region}.
	
	Suppose there exists a vertex $ v\in V\backslash I_k $. Then $ v\not\in I $. If $ v\in V_0 $, then there is no edge from $ v $ into $ I_k $, else $ v $ would be in $ B $ in the $ k+1 $-th iteration, and $ I_k\not= I_{k+1} $. If $ v\in V_1 $, then there exists an edge from $ v $ to $  V\backslash I_k  $, else $ v $ would be in $ A $ in the $ k+1 $-th iteration again. 
	
	Then there is no strategy for $ \pz $ to visit $ I_k $, and $ I $ in particular, from $ v $, implying $ v\not\in \solveBuchi(\gamegraph, I) $, which would be a contradiction to the fact that every vertex is \buchi winning for $ \pz $.
	
	Then $ I_k=I_{k+1}=V $. Hence the while loop will be exited, and the procedure, and hence the algorithm, terminates.
	
	\paragraph*{Complexity analysis:} The procedure \solveBuchi takes time $ \bigO(nm) $. Then the while loop in the \reachTemp ~procedure has at most $ n $ many iterations since at least one vertex is added to $ I $ in each iteration. Each iteration take $ \bigO(m) $ time, resulting in the total complexity of $ \bigO(nm) $ for the procedure \reachTemp, and hence, for the algorithm.
\end{claimproof}

With this, we are ready to prove soundness of the constructed template.	Let $ \stratz $ be a strategy following $ \assump $, and let $ \play=v_0\ldots v_i\ldots $ be a play compliant with $ \stratz $ originating at $ v_0\in \wino $.
	
	We first note that the play never leaves the winning region due to the safety part of the template. Now let $ k\in \N $ be such that in the proof of the lemma above, and $ A_i $ and $ B_i $ be the values of $ A $ and $ B $ in the $ i $-th iteration of the while loop in the algorithm above, i.e. $ I_i=A_i\cup B_i $ for $ 1\leq i\leq k $. 
	
	We show that $ \play $ visits $ I=I_0 $ infinitely often. To this end, let $ \gamma=c_0\ldots c_i\ldots \in [0;k]^{\omega} $, such that $ c_i=min\{j\in[0,k]\mid v_i\in I_j\} $, i.e. $ c_i $ is the iteration of while loop in the \reachTemp ~procedure when $ v_i $ was added in $ I $. We show that $ 0 $ occurs infinitely often in $ \gamma $. Suppose not, i.e. the minimum number $ m $ occurring infinitely often in $ \gamma $ is not $ 0 $. Then vertices from $ I_m\backslash I_{m-1} $ occur infinitely often in $ \play $. Then if vertices from $ A_m\backslash I_{m-1} $ are visited infinitely often, then since both players are forced to visit $ I_{m-1} $ every time $  A_m\backslash I_{m-1} $ is visited, by the definition of $ \textsf{uattr} $, we get a contradiction. If vertices from $ B_m $ are visited infinitely often, then since $ \stratz $ follows $ \assumpgrlive(\livegroup) $, infinitely often edges leading towards $ A_m $ are taken, again giving rise to a contradiction. Hence, $ m=0 $, proving that $ \play $ is winning for $ \pz $, and $ \assump $ is a winning strategy template for vertices in $ \wino $.
\end{proof}

\subsection{Strategy Templates for \cobuchi Games}\label{app:coBuechi}
Here, we restate \cref{thm:coBuechiCont}, and formally prove the same. 
\restatecobuechi*

\begin{proof}
Before proceeding with the proof of \cref{thm:coBuechiCont} we show that \cref{alg:coBuchitemplate} indeed terminates.
\begin{lemma}\label{lemma:cobuchi algo terminates}
	The \cref{alg:coBuchitemplate} terminates in time $ \bigO(nm) $, where $ n $ and $ m $ are as usual.
\end{lemma}
\begin{claimproof}
	We first note that the inner while loop (\cref{step:inner while loop begins}-\ref{step:inner while loop exits}) terminates. This is simple to observe since $ A $ only grows and since there are finitely many vertices the termination condition will be satisfied eventually.
	
	We need to show that every vertex $ v\in V $ gets added to $ A $ in some iteration of the outer while loop. 
	
	We prove by induction that after every iteration of the outer loop, every vertex in the remaining graph is still \cobuchi winning in the remaining graph. The base case is trivial, because we recall that every vertex of $ \gamegraph $ is \cobuchi winning, due to \cref{step:reduce to cobuchi region}.
	
	We denote the graph after the $ i $-th iteration by $ \gamegraph_i =(V_i,E_i)$. Let the statement above holds true after $ i $-th iteration, i.e. $ V_i=\solveCobuchi(\gamegraph_i, I_i) $. Let $ u\in V_{i+1} $. For the $ i+1 $-th iteration $ A_{i+1}=\solveSafety(\gamegraph_{i}, I_i) $. Now if $ \pz $ had a strategy of reaching $ A_{i+1} $ in $ \gamegraph_i $, then it would be included in $ A_{i+1} $ after the inner while loop is executed. But since this is not the case, and $ u $ is still winning in $ G_i $, then the winning strategy is such that the plays do not necessarily stay in $ A_{i+1} $. Hence, even if $ A_{i+1} $ is removed from the graph, $ u $ is still winning in $ \gamegraph_{i+1} $ with the same winning strategy.
	
	Now for a vertex $ u $ to be \cobuchi winning, there exists a strategy such that every play starting at $ u $ eventually ends up in a subset $ J $ of $ I $. But $ I $ gets strictly smaller in every iteration of the outer while loop, and it can happen only finitely often. Hence if $ V $ never reduces to $ \emptyset $, there is a winning vertex $ v\in V $ but there is no $ J\subseteq I $, where the play starting at $ v $ can eventually end up in, producing a contradiction.
	
	\paragraph*{Complexity analysis:} The \solveCobuchi ~procedure takes $\bigO(nm) $ time. Then the outer loop needs at most $ n $ iterations, and the inner loop needs at most $ \bigO(m) $ time, since it is just the $ \mathsf{attr} $ computation, resulting in total complexity of $ \bigO(nm) $ for the algorithm.
\end{claimproof}

	Now let $ \stratz $ be a strategy following $ \assump $, and let $ \play=v_0\ldots v_i\ldots $ be a play compliant with $ \stratz $ originating at $ v_0\in \wino $.
	
	Let $ V_i=A_i\cup V_{i-1} $ and $ V_1=A_1 $. Intuitively, $ V_i $ is the set of vertices which have been removed from the initial graph after $ i $-th iteration of the outer while loop.	We denote by $ \gamegraph'_i $ the restricted graph $ \gamegraph|_{V_{i}} $. 
	
	We first show that if a play eventually stays in $ \gamegraph'_i $ then it is winning for $ \pz $. For the base case, when $ i=1 $, this is easy to see: because the play can go further away from $ A_{i} $, only finitely often due to the co-live edges added in \cref{step:colive to force going toward safety winning vertices}, and eventually the play stays in $ A_1\subseteq I $. Hence the play would be \cobuchi winning.
	
	Now let the statement holds true for $ \gamegraph'_{i-1} $ for some $ i-1\in \N $. Now if the play stays in $ \gamegraph'_i $. Then if the play stays in $ A_{i} $ then it is winning by the arguments similar to the base case. Else it will eventually end up in $ V_{i-1} $, since it can not go to $ A_i $ infinitely often from $ V_{i-1} $ due to the co-live edges added in \cref{step:colive to force going toward safety winning vertices} in the last iteration of inner while loop and \cref{step:colive to force staying in safety winning vertices}. Then by the induction hypothesis, it is again winning.
	
	Hence, by induction the statement holds true for every $ i $, and in particular for $ k $, where $ k $ is the total number of iterations of the outer while loop. Since due to the safety part of the template, the play $ \play $ stays in $ \gamegraph_k $, and hence is \cobuchi winning. Hence, $\assump $ is a $ \pz $ winning template for vertices in $ \wino $.
\end{proof}

\subsection{Strategy Templates for Parity Games}\label{app:parity}
We formally show that the strategy template constructed using \cref{alg:parityTemplateZeilonka} is winning for $ \pz $. We restate \cref{thm:ParityCont} for convenience. 
\restateparity*
\begin{proof}
	Before we prove that \cref{alg:parityTemplateZeilonka} gives a winning strategy template, we show that it terminates.
	\begin{lemma}
		The \cref{alg:parityTemplateZeilonka} terminates in time $ \bigO(n^{d+\bigO(1)}) $.
	\end{lemma}
	\begin{claimproof}
		This is fairly easy to see since this is a simple modification of the usual Zeilonka's algorithm for parity games, and the call to the \reachTemp ~procedure terminates as shown in previous section. The complexity can be obtained by the usual analysis for Zeilonka's algorithm.
	\end{claimproof}
	
	We now prove that \cref{thm:ParityCont} gives a winning strategy template.
	
	Let $ \stratz $ be a strategy following $ \assump $, and let $ \play=v_0\ldots v_i\ldots $ be a play compliant with $ \stratz $ originating at $ v_0\in \wino $. 

	We prove by induction on the number of vertices $ n $ in $ G $ that $\stratz $ is winning. When $ n=1 $, this is trivially true. Now suppose that the statement holds for graphs of size $ k $. Now let $ n=k+1 $, and $ d $ be the highest priority occurring in $ \gamegraph $. First, we notice that $ \stratz $ does not allow $ \play $ to visit $ \winl $ by the correctness of safety templates.
	
	Now, if $ d $ is odd. Note that if $ \play $ visits $ \wino $ infinitely often, it will eventually stay in $ \wino $ due to co-live edges added in \cref{alg:step:noleavewo}, then by induction hypothesis, $ \play $ satisfies the parity winning condition. Else $ \play $ eventually stays in $ \wino' $, since if it goes to $ B\backslash\wino' $ infinitely often, then it will again visit $ \wino $ infinitely often due to live-groups added in \cref{alg:step:reachwo} and we can argue as above. Again $ \play  $ will be winning by induction hypothesis, if it stays in $ \wino' $.
	
	Otherwise, if $ d $ is even. If the play visits $ A $ infinitely often, then $ P_d $ is visited infinitely often due to the live-groups added in the \cref{alg:step:reachhigheven}. Otherwise, by induction hypothesis, if $ \play $ stays in $ \wino' $, it is winning again.
	
	Hence, by induction, $ \stratz $ is a winning strategy for $ \pz $, implying that $ \assump $ is a winning strategy template. 
\end{proof}

\subsection{Extracting Strategies from Strategy Templates}\label{app:prop:implementability}
We show that \cref{prop:implementability} holds, i.e., that Algorithms~\ref{alg:buchitemplate}, \ref{alg:coBuchitemplate} and \ref{alg:parityTemplateZeilonka} 
 always return conflict-free templates. 
\restateImplementability*
 
 \begin{proof}
The claim directly follows from the definition of the algorithms in the following way. 
%
First note that in every of the three algorithms, we only have unsafe edges going \emph{out of} the winning region and all other restrictions are on the edges \emph{inside} the winning region. Hence, there cannot be any conflict involving unsafe edges. Moreover, since the template returned by $\buchiTemp$ does not contain any co-live edge, it is easy to see that for such templates (i) and (ii) in \cref{def:implementability} can never occur. Furthermore, the algorithm for $\coBuchiTemp$ only adds a co-live edge when there is some other choice from the source vertex, showing that (i) in \cref{def:implementability} cannot occur. Moreover, there are no live-groups in the templates returned by $\coBuchiTemp$, hence (ii) in \cref{def:implementability} cannot occur.
%
Similarly, in the algorithm for $\parityTemp$, i.e., \cref{alg:parityTemplateZeilonka}, we only add co-live edges in \cref{alg:step:noleavewo}, which are going out of $\wino$, where $\wino$ is the winning region in a restricted game graph. Hence, there is always another choice from source vertices of such edges. Moreover, the live-groups it computes in \cref{alg:step:reachwo} contain edges which are inside $\wino$; and the live-groups computed in \cref{alg:step:reachhigheven} can never contain a co-live edge (as in that part of the algorithm we do not add any co-live edge). Therefore, the template returned by \cref{alg:parityTemplateZeilonka} is also conflict-free as cases (i) and (ii) of \cref{def:implementability} cannot occur.
\end{proof}

\section{Applications of Strategy Templates}

\subsection{Proof of \cref{lem:compositon}}\label{app:reductionToSingleParity}

\reductionToSingleParity*

\begin{proof}
	First of all, note that $P'_i =  P_i\setminus U$ for every $i\leq 2d$ and $P'_{2d+1} = P_{2d+1}\cup U$ by construction.
	Now, let us start by showing that $\lang(\spec') \subseteq \lang(\spec \wedge \lozenge\square\neg U)$.
	Suppose $\play \in \lang(\spec')$. Then for some priority $2j$, the play $\play$ visits $P'_{2j} \subseteq P_{2j}$ infinitely often and $\inf(\play) \subseteq \bigcup_{i\leq 2j} P'_{j}$, which implies $\inf(\play) \subseteq \bigcup_{i\leq 2j} P_{j}$. Hence, $\play \in \lang(\spec)$. Furthermore, as $\inf(\play) \cap P'_{2d+1} = \emptyset$ (since $ \play $ satisfies parity condition) and $U\subseteq P'_{2d+1}$, it holds that $\inf(\play) \cap U = \emptyset$. Hence, $\play \in \lang(\lozenge\square\neg U)$. Therefore, $\play \in  \lang(\spec \wedge \lozenge\square\neg U)$. The other direction follows similarly. Hence, $\lang(\spec') = \lang(\spec \wedge \lozenge\square\neg U)$.
	
	Note that by the above result, it holds that $\lang(\spec')\subseteq\lang(\spec)$. Now, if a $\p{0}$'s strategy $\strat$  is winning from a vertex $u$ in game $\game'$. Then it holds that $\lang(u,\strat) \subseteq \lang(\spec')$, which implies $\lang(u,\strat) \subseteq \lang(\spec)$. Hence, $\strat$ is also a winning strategy from $u$ in $\game$.
	Now, suppose a strategy template $\assump$ is winning from $u$ in $\game'$. Then every strategy satisfying the template $\assump$ is winning from $u$ in $\game'$, and hence, is winning from $u$ in $\game$. Therefore, the template $\assump$ is also winning from $u$ in $\game$.
\end{proof}

\subsection{Correctness of \compose}\label{app:correctnessOfCompositionality}
We recall \cref{thm:compositionCorrectness}, and prove the correctness and {conflict-freeness} of the strategy templates obtained by \compose ~here.
\correctnessOfCompositionality*

\begin{proof}
	Let us denote $P_{1,2d+1}$ to be the set of vertices $v$ such that $\priority_1(v) = 2d_1+1$.
	We show every claim using induction on the pair $(\abs{\wino'},\abs{\wino'\setminus P_{1,2d+1}})$ (ordered lexicographically), where $\wino'$ is the vertex set taken as input in the algorithm. As in the theorem statement, initially we have $\wino' = V$. 
	
	For base case, if $\abs{\wino'} = 0$, then $\abs{V} = 0$ and hence, it returns $\wino=\emptyset$; and if $\abs{\wino'}-\abs{P_{1,2d+1}} = 0$, then it is easy to see that $\wino = \emptyset$. Hence, if $\conflict_1 \cup \conflict_2 = \emptyset$, then it returns $\wino=\emptyset$. Otherwise, $\gamegraph|_{\wino}$ is an empty game graph; hence, in the next iteration, we have $\wino' = \emptyset$. So, each $W_i$ and each strategy template is empty. Hence, $\conflict_1 \cup \conflict_2 = \emptyset$ holds (in the next iteration), and it returns $\wino=\emptyset$. So, in any case, it returns an empty set as $\wino$, and an {conflict-free} strategy template that is trivially winning from $\wino$.
	
	Now for the induction case, suppose $\abs{\wino'}$ and $\abs{\wino'\setminus P_{1,2d+1}}$ are positive. 
	It is easy to verify that $\conflict_1$ and $\conflict_2$ corresponds to the set of conflicted vertices due to the condition (i) and (ii), respectively of \cref{def:implementability}.
	Hence, if $\conflict_1 \cup \conflict_2 = \emptyset$, then there is no conflict in the conjunction of the strategy templates. Hence, conjuncting winning strategy templates for all games $(\gamegraph, \spec_i)$ actually gives us an {conflict-free} and winning strategy template for the game $\bigwedge_{i\leq k}\spec_i$ as any strategy satisfying the strategy templates of every game is winning in every game. Therefore, the complete winning region for the game $(\gamegraph,\bigwedge_{i\leq k}\spec_i)$ is the intersection of the winning regions $W_i$. Hence, the algorithm returns the correct winning region and a winning strategy template for the game  $(\gamegraph,\bigwedge_{i\leq k}\spec_i)$.
	
	Now, if $\conflict_1 \cup \conflict_2 \not = \emptyset$, then some vertices are added to $P_{1,2d+1}$ (\cref{alg:compositon:line8}) for the next iteration.  Note that $\wino \subseteq \wino'$ and $\conflict_1\cup \conflict_2 \subseteq \wino'$ as the parity games are solved in \cref{alg:composition:parity} are solved for the game graph restricted to $\wino'$. So, in every iteration, $\wino'$ stays unchanged or gets smaller.
	Let $\wino''$ and $P'_{1,2d+1}$ are the $\wino'$ and $P_{1,2d+1}$ in the next iteration.
	If $\wino'$ gets smaller in the next iteration, then $\abs{\wino''} < \abs{\wino'}$. Else, we have $\abs{\wino''} = \abs{\wino'}$ and $\abs{\wino' \cap P'_{1,2d+1}} > \abs{\wino' \cap P_{1,2d+1}}$, which implies $\abs{\wino''\setminus P'_{1,2d+1}} < \abs{\wino'\setminus P_{1,2d+1}}$. 
	Hence, in any case, 
	\[(\abs{\wino''},\abs{\wino''\setminus P'_{1,2d+1}}) <_{lex} (\abs{\wino'},\abs{\wino'\setminus P_{1,2d+1}}).\]
	Then, by the induction hypothesis, the strategy template $\assump$ returned by the algorithm is {conflict-free} and winning from every vertex $u\in \wino$ in game $(\gamegraph,\bigwedge_{i\leq k}\spec'_i)$. 
	It is enough to show that the strategy template $\assump$ is also winning from every vertex $u\in \wino$ in game $(\gamegraph,\bigwedge_{i\leq k}\spec_i)$. 
	
	As $W_i$ is the winning region for the game $(\gamegraph,\spec_i)$ for each $i$, the winning region for the objective $\bigwedge_{i\leq k} \spec_i$ is a subset of $\bigcap_{i\leq k}W_i$. It is easy to see that a winning strategy is still winning if we restrict the game graph to the winning region. Furthermore, by \cref{lem:compositon}, it holds that the strategy template $\assump$ is indeed winning from every vertex $u\in \wino$ in game $(\gamegraph,\bigwedge_{i\leq k}\spec_i)$.
	
	For the complexity analysis, as the pair $(\abs{\wino'},\abs{\wino'\setminus P_{1,2d+1}})$ can decrease at most $n^2$ times, the maximum number of iterations is $n^2$. In each iteration, the algorithm call $\parityTemp$ for each game $(\gamegraph,\spec_i)$ which runs in $\bigO(n^{2d_i+1})$ time, and some additional operations which take polynomial time in number of edges. Hence, in total, the algorithm runs in time $\bigO(kn^{2d+3})$, where $d = \max_{i\leq k} d_i$.
\end{proof}

\subsection{Proof of \cref{prop:gafcorrection}}\label{sec:gafcorrection}





\GAFcorrection*

\begin{proof}
Suppose that for every vertex $ v\in V^0 $, there is an outgoing edge which is not in $ \safegroup\cup\colivegroup\cup \faulty $. Then consider the strategy $ \strat_g $ that, at time $ t $, takes the edge $ e\in E_t\backslash(\safegroup\cup\colivegroup) $ for all vertices without any live-group, and for the ones with live-groups, it alternates between the edges satisfying the live-groups whenever they are available, and the edge $ e\in E_t\backslash(\safegroup\cup\colivegroup) $ when no live-group edge is available. We show that $ \strat_g $ is winning for $ \pz $, by showing that it is compliant with $ \assump $, and invoking \cref{thm:ParityCont}. It is easy to observe that $ \strat_g $ is compliant with the safely and co-liveness part of $ \assump $. Now, to be compliant with the live-group template, we observe that the group-live edges will be available infinitely often when the play visits the source vertex, so $ \strat_g $ will indeed choose the edges alternately. Hence it will be compliant with the strategy template $ \assump $. 
\end{proof}


\section{Experimental Results}\label{app:experiments}

For completeness, we show a version of \cref{fig:intro-experiement-sensitivity} including all solvers considered in \cref{table:experiment} in \cref{fig:experiment-sensitivity-all}.

\begin{figure}[t]
    \includegraphics[scale=0.29]{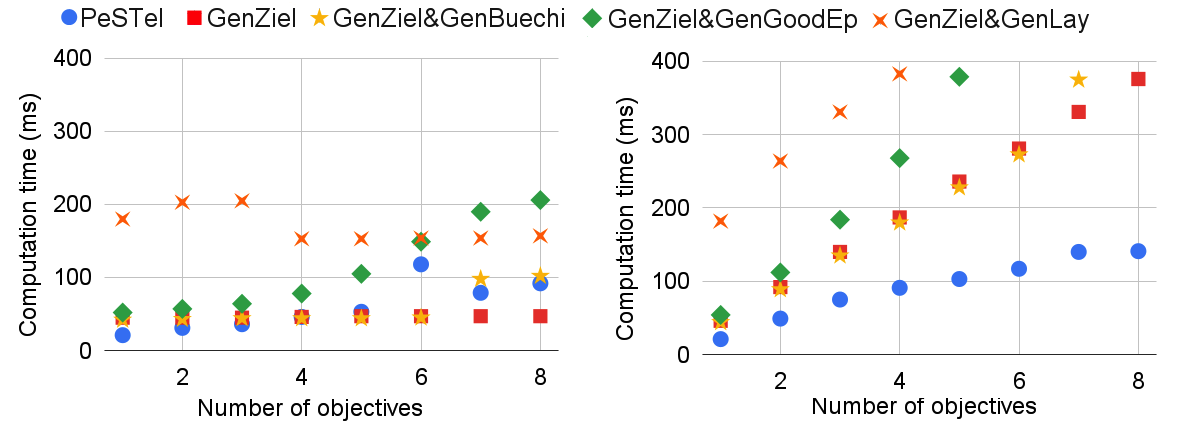}
    \vspace{-0.2cm}
     \caption{Experimental results over 1400 benchmark instances showing the sensitivity of different tools on the number of objectives. Data points give the average execution time (in ms) over all instances with the same number of objectives. Left: all objectives are given upfront. Right: objectives are added one by one.}
     \label{fig:experiment-sensitivity-all}
     \vspace{-0.5cm}
 \end{figure}

\end{document}